\documentclass[12pt,oneside]{article}

\usepackage{cite}
\usepackage{acronym}
\usepackage{enumerate}  
\usepackage{a4wide}
\usepackage[left=2.5cm, right=2.5cm, top=2.5cm, bottom=2.5cm]{geometry}
\usepackage{fancyhdr}
\usepackage{graphicx}
\usepackage{palatino}
\usepackage{blindtext}
\usepackage{multirow}
\usepackage{amsmath}
\usepackage{amssymb}
\usepackage{bbm}
\usepackage{ragged2e}
\usepackage{xcolor}
\usepackage{mathrsfs}

\usepackage[T1]{fontenc}
\usepackage[utf8]{inputenc}
\usepackage[bookmarks]{hyperref}
\usepackage[justification=centering]{caption}
\usepackage{csquotes}
\usepackage[ngerman]{cleveref}

\usepackage{tikz}
\usepackage{tikz-cd}
\usetikzlibrary{calc}
\usetikzlibrary{decorations.pathmorphing}
\tikzset{snake it/.style={decorate, decoration=snake}}

\newtheorem{theorem}{Theorem}[section]
\newtheorem{lemma}[theorem]{Lemma}
\newtheorem{definition}[theorem]{Definition}
\newtheorem{proposition}[theorem]{Proposition}
\newtheorem{corollary}[theorem]{Corollary}

\newcommand{\qedsymbol}{\hfill\square}
\definecolor{dgreen}{RGB}{17, 122, 89} 
\definecolor{ggreen}{RGB}{51, 173, 95} 

\title{Noncommutative QFT and Relative Entropy \\ on Axisymmetric Bifurcate Killing Horizons}
\author{Philipp Dorau$^{1,}$\footnote{\href{mailto:philipp.dorau@uni-leipzig.de}{philipp.dorau@uni-leipzig.de}}, $\quad\quad$ Albert Much$^{2,}$\footnote{\href{mailto:amuch@sharjah.ac.ae}{amuch@sharjah.ac.ae}}, $\quad\quad$ Rainer Verch$^{1,}$\footnote{\href{mailto:rainer.verch@uni-leipzig.de}{rainer.verch@uni-leipzig.de}}\\[0.42cm]
\begin{minipage}{0.92\textwidth}
\raggedright\small\textit{$^1$Institut f\"ur Theoretische Physik, Universit\"at Leipzig, Br\"uderstr. 16, 04103 Leipzig, Germany}\\[0.12cm]
\textit{$^2$Department of Applied Physics and Astronomy, College of Sciences, University of Sharjah,\\
$\,\,\,$Sharjah, United Arab Emirates}
\end{minipage}}
\date{October 1, 2026}

\begin{document}

\maketitle

\begin{abstract}
We construct a deformed algebraic quantum field theory on bifurcate Killing horizons in stationary axisymmetric spacetimes. The deformation is generated by the commuting actions of affine dilations along the null generators of the horizon and rotations about the axis of symmetry, analogously to the Moyal-Rieffel deformation. Physically, this effectively implements a noncommutative geometric structure of the horizon. Moreover, we compute the relative entropy between coherent states in the deformed horizon theory, which remains strictly positive and exhibits a novel second-order correction in the deformation parameter, which becomes particularly significant for black holes whose horizon area is sufficiently small for Planck-scale effects to become non-negligible.
\end{abstract}

\section{Introduction}\label{Introduction}
There are several arguments that the classical description of spacetime as a smooth manifold is unlikely to remain valid at arbitrarily small length scales. Most prominently, operational limitations on spacetime localization arise when Heisenberg’s quantum mechanical uncertainty principle is combined with gravitational collapse arguments. Attempts to localize events at distances below the Planck length inevitably require energy concentrations that would classically lead to the formation of black holes, thereby obstructing any further localization \cite{DFR:1994lett,DFR:1995ncs}. Various approaches to quantum spacetime models and quantum gravity address this problem, one particular possibility being given by intrinsically noncommutative models of spacetime, in which position operators are subject to nontrivial commutation relations \cite{DFR:1995ncs,Much:2025genrd}. This theoretical observation naturally raises the question of how quantum field theory (QFT) should be formulated on such noncommutative backgrounds. \\

Within the operator algebraic formulation of QFT, noncommutative background geometries are typically implemented not through explicit modifications of the underlying spacetime manifold, but rather via strict or formal deformations of the respective algebra of observables \cite{BS:2008wcn,BLS:2011wcrd,DLMM:2011def,Much:2021mdga,Much:2025genrd}. In this setting, star product deformations provide a mathematically rigorous way to incorporate noncommutative background structures directly at the level of the algebra of observables, while simultaneously preserving some fundamental properties of the original theory and introducing well-defined Planck-scale correction terms. From the viewpoint of QFT on curved spacetimes, backgrounds with nontrivial symmetry groups, particularly those with commuting Killing vector fields, provide a natural arena in which such deformations can be constructed rigorously and related directly to the underlying geometry, see, e.g., \cite{DLMM:2011def}. This observation motivates the study of symmetry-adapted deformations of QFTs as a conservative yet fruitful approach to probing possible effects of noncommutative spacetime models, with particular emphasis on geometric horizons, which play a distinguished role due to their universal geometric and algebraic properties. \\

The aim of this work is to construct and analyze a deformation procedure for QFT on bifurcate Killing horizons in stationary and axisymmetric spacetimes. The deformation is generated by a pair of commuting operators, namely dilations along the null generators of the horizon, which correspond to the projected Killing flow onto the horizon, and by the infinitesimal generator of rotations about the axis of symmetry. These two commuting spacetime symmetries reflect the geometric structure of axisymmetric bifurcate Killing horizons and guarantee the associativity of the resulting star product. Moreover, the deformation effectively implements a noncommutative horizon geometry, intertwining null and angular coordinates. This provides a mathematically precise symmetry-adapted framework to probe possible noncommutative corrections to near-horizon QFT without imposing additional assumptions beyond those already given by the classical background geometry. \\

Our approach fits into the broader programme of constructing deformations of quantum field theories on curved spacetimes with nontrivial isometry groups, primarily following the framework developed in \cite{DLMM:2011def}. In that work, warped convolution techniques were employed to deform free field theories using the action of commuting Killing vector fields. The present analysis is closely related in spirit, but departs in a few significant aspects by focusing exclusively on the null geometry of a bifurcate Killing horizon and by implementing the deformation at the level of the underlying symplectic space rather than by deforming the quantum fields directly. While both approaches rely on commuting spacetime symmetries to ensure associativity, the symplectic deformation perspective adopted here is particularly well suited to horizon settings, where the symplectic structure plays a fundamental role (cf. \cite{KPV:2021ea}). In our case, the deformation is specifically generated by the commuting symmetries of affine dilations along the horizon generators and rotations about the spacetime's axis of symmetry, which together form a two-dimensional Abelian symmetry group and provide a natural starting point for a Rieffel-type star product. In this sense, the horizon takes a role analogous to that played by Minkowski spacetime in previously established deformation schemes, see, e.g., \cite{BLS:2011wcrd}, albeit with an adapted choice of commuting spacetime symmetries specifically tailored to stationary black hole horizons. \\

The main results of this paper are as follows. Starting from a quantized free real scalar field theory on a stationary axisymmetric spacetime with a bifurcate Killing horizon, we consider its symmetry-improving restriction to the horizon in the sense of Summers and Verch \cite{SV:1996tomitaki}, whose structural features are reviewed in Section \ref{QFTonKillingHorizons}. In Section \ref{ConstructionAxisymmetricProduct}, we construct an algebraic deformation procedure for this theory on the horizon, based on the commuting actions of affine dilations along the horizon generators and rotations about the axis of symmetry. We then derive some fundamental properties of the associated star product on a suitable space of compactly supported functions on the horizon.  While the construction itself is fully explicit and nonperturbative, perturbative expansions nevertheless play an essential role, allowing associativity to be established in the sense of formal power series, and enabling a clear physical interpretation of the resulting noncommutative corrections. Subsequently, in Section \ref{SecDeformedQFT}, we employ our novel star product to define a deformed symplectic space and its associated Weyl algebra, thereby obtaining a deformed QFT on the Killing horizon. Within this framework, we compute the relative entropy between coherent states in the deformed theory in Section \ref{SecDeformedRelativeEntropy}, obtaining both a general expression and a perturbative expansion up to second order in the deformation parameter. In particular, the second-order contribution turns out to be strictly nonnegative, ensuring positivity of the relative entropy in the deformed theory. This provides a concrete illustration of how a noncommutative background geometry affects quantum information theoretic quantities and, in particular, yields upward corrections to the Page curve (cf. \cite{Page:1993og,Page:2013rev}) for the relative entropy, which reflect Planck-scale modifications of black hole thermodynamics.

\section{Preliminaries}\label{Preliminaries}

\subsection{Classical Spacetime Geometry}
In this work, we focus on globally hyperbolic spacetimes $(\mathcal{M},g)$ which are stationary and axisymmetric, i.e., spacetimes that possess a timelike Killing vector field $\xi^a$ and a spacelike Killing vector field $\psi^a$ whose integral curves are closed. It is important to notice that the two Killing vector fields $\xi^a, \psi^a$, or equivalently, the actions of their corresponding isometry groups, commute in any such spacetime \cite{Carter:1970kc}. \\

Moreover, we restrict ourselves to stationary and axisymmetric spacetimes that possess a bifurcate Killing horizon, as defined in \cite{Carter:1969kh,Boyer:1969kh,KW:1991hadamard}. More specifically, a Killing horizon is a null hypersurface $\mathcal{H} \subset \mathcal{M}$, such that a Killing vector field $k^a$ is normal to $\mathcal{H}$ and null on $\mathcal{H}$, i.e.,
\begin{equation}
k^a k_a = 0.
\end{equation}

\noindent In particular, the special case of a bifurcate Killing horizon arises when the spacetime $(\mathcal{M},g)$ possesses a $2$-dimensional spacelike submanifold $\mathcal{S}$, called the bifurcation surface or alternatively horizon cross-section, such that the out- and ingoing null congruences\footnote{A null congruence emanating from $\mathcal{S}\subset\mathcal{M}$ is a family of non-intersecting curves in $\mathcal{M}$, whose tangent vectors are null, such that for each $p\in\mathcal{S}$, exactly one curve of this family passes through $p$ \cite{Poisson:2004tool}.} $\mathcal{H}_A$ and $\mathcal{H}_B$, emanating from $\mathcal{S}$, are both Killing horizons with respect to $k^a$ \cite{Boyer:1969kh,KW:1991hadamard}. Moreover, the structure of a bifurcate Killing horizon divides the spacetime $(\mathcal{M},g)$ into four different regions, namely the left wedge $\mathcal{L}$, the right wedge $\mathcal{R}$, as well as the future and past wedges $\mathcal{F}$ and $\mathcal{P}$, as depicted in Figure \ref{FigureBifurcateKillingHorizon}, which are also often referred to as black hole and white hole regions, respectively.

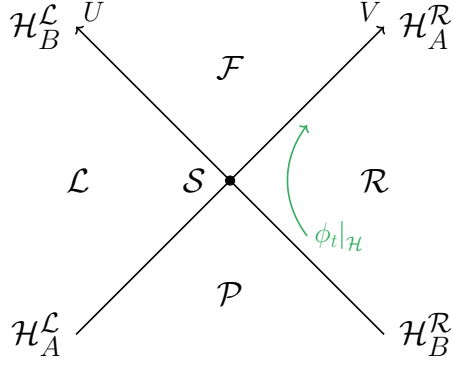
\begin{figure}[h!]
\centering
\resizebox{0.37\textwidth}{!}{
\begin{tikzpicture}
  % Define coordinates
\coordinate (topleft) at (-4,4);
\coordinate (bottomleft) at (-4,-4);
\coordinate (topright) at (4,4);
\coordinate (bottomright) at (4,-4);
\coordinate (origin) at (0,0);

\coordinate (midleft) at (-3.5,0);
\coordinate (midright) at (3,0);
\coordinate (midtop) at (0,2.5);
\coordinate (midbottom) at (0,-2.5);

\coordinate (bhtl) at (-7,0.5);
\coordinate (bhtr) at (-6,0.5);
\coordinate (bhbl) at (-7,-0.5);
\coordinate (bhbr) at (-6,-0.5);
\coordinate (bhc) at (-6.5,0);

\coordinate (boost1) at (2, -1.44);
\coordinate (boost2) at (2, 1.44);

% Draw UV axis
  \draw[->, very thick] (bottomright) -- (topleft);
  \draw[->, very thick] (bottomleft) -- (topright);

%Draw "Boost"
\draw[->, very thick, ggreen] (boost1) to[out=127,in=-127] (boost2);

% Draw BH region indicator
%\draw (bhbl) -- (bhtr);
%\draw (bhbr) -- (bhtl);
%\draw (bhtl) -- (bhtr);
%\draw (bhbl) -- (bhbr);
%\fill[black, opacity=1] (bhc) -- (bhtr) -- (bhtl) -- cycle;

% Draw labels
\node[left, font=\huge] at (origin) {$\mathcal{S}\;\;$};
\node[left, font=\huge] at (topleft) {$\mathcal{H}^\mathcal{L}_B\;$};
\node[right, font=\huge] at (topright) {$\;\mathcal{H}^\mathcal{R}_A$};
\node[above, font=\LARGE] at (topleft) {$\quad\; U$};
\node[above, font=\LARGE] at (topright) {$V\quad$};
\node[left, font=\huge] at (bottomleft) {$\mathcal{H}^\mathcal{L}_A\;$};
\node[right, font=\huge] at (bottomright) {$\;\mathcal{H}^\mathcal{R}_B$};
\node[right, ggreen, font=\LARGE] at (boost1) {$\left. \phi_t \right\vert_\mathcal{H}$};
\node[left, font=\huge] at (midleft) {$\mathcal{L}$};
\node[right, font=\huge] at (midright) {$\;\mathcal{R}$};
\node[above, font=\huge] at (midtop) {$\mathcal{F}$};
\node[below, font=\huge] at (midbottom) {$\mathcal{P}$};

\draw[black,fill=black] (origin) circle (.6ex);
  
\end{tikzpicture}
}
\caption{\justifying \small{Structure of a bifurcate Killing horizon, consisting of the null hypersurfaces $\mathcal{H}_A, \mathcal{H}_B$, intersecting at the horizon cross-section $\mathcal{S}$, where $\mathcal{H}_A$ is affinely parametrized by $V$, and similarly $\mathcal{H}_B$ by $U$. The spacetime is separated into the wedge-shaped regions $\mathcal{L}, \mathcal{R}, \mathcal{F},\mathcal{P}$, in particular, dividing and the Killing horizons $\mathcal{H}_A, \mathcal{H}_B$ into the parts $\mathcal{H}_A^\mathcal{R},\mathcal{H}_A^\mathcal{L}$ and $\mathcal{H}_B^\mathcal{L},\mathcal{H}_B^\mathcal{R}$, respectively. The green line depicts the Killing flow $\left. \phi_t \right\vert_\mathcal{H}$ generated by $\xi^a$, projected to $\mathcal{H}^\mathcal{R}$.}}
\label{FigureBifurcateKillingHorizon}
\end{figure}

From a physical perspective, bifurcate Killing horizons represent a mathematically idealized and generalized notion of black hole horizons, capturing the essential features of the causal structure of stationary black holes. Prominent examples are provided by the Schwarzschild and Kerr solutions of the vacuum Einstein equations \cite{Faraoni:2015horizons}. \\

An important quantity associated to a bifurcate Killing horizon is the surface gravity $\kappa$, which for a Killing horizon with respect to the Killing vector field $k^a$, is defined by the relation
\begin{equation}
\label{DefinitionKillingSurfaceGravity}
k^a \nabla_a k^b =: \kappa \, k^b
\end{equation}

\noindent on $\mathcal{H}$ \cite{Wald:1984gr}. In other words, the surface gravity $\kappa$ is a measure for how geodesics with tangent vector $k^a$ fail to be affinely parametrized on the Killing horizon $\mathcal{H}$ \cite{Faraoni:2015horizons}. \\

Lastly, we want to summarize the properties of the flow $\left. \phi_t \right\vert_\mathcal{H}$ generated by the timelike Killing vector field, projected to the Killing horizon $\mathcal{H}$, as sketched in Figure \ref{FigureBifurcateKillingHorizon}. Assuming that the affine parameters $U,V$ of the Killing horizons $\mathcal{H}_B, \mathcal{H}_A$ are chosen in such a way that points on the horizon cross-section $\mathcal{S}$ correspond to $U = 0 = V$, it was shown by Summers and Verch \cite{SV:1996tomitaki} that the projected flow $\left. \phi_t \right\vert_\mathcal{H}$ acts as affine dilations $\mathfrak{D}_t$ along $\mathcal{H}_A$ and $\mathcal{H}_B$, i.e.,
\begin{equation}
\phi_t (V,p) = \mathfrak{D}_{\kappa t} (V,p) := (e^{\kappa t} V , p) \quad \mathrm{and} \quad \phi_t(U,p) = \mathfrak{D}_{-\kappa t} (U,p) := (e^{-\kappa t} U, p)
\end{equation}

\noindent for $t\in\mathbb{R}$, and points $(V,p)\in\mathcal{H}_A$ and $(U,p)\in\mathcal{H}_B$\footnote{We remark that our convention for the affine parameters on the horizons is opposite to that presented in \cite{KW:1991hadamard} and \cite{SV:1996tomitaki}, which assign $U$ to $\mathcal{H}_A$ and $V$ to $\mathcal{H}_B$. Instead, we follow the convention used in \cite{KPV:2021ea}.}, respectively. We also refer to \cite{KPV:2021ea} for a generalization of this behaviour to nonstationary spherically symmetric spacetimes, using the Kodama flow instead of the Killing flow. \\

In the remainder of this paper, we specifically restrict ourselves to the Killing horizon $\mathcal{H}_A$, with the understanding that all considerations analogously hold for $\mathcal{H}_B$, up to possible sign changes \cite{SV:1996tomitaki,KPV:2021ea}. Moreover, since we restrict ourselves to axisymmetric spacetimes, we henceforth parametrize points $p\in\mathcal{S}$ by coordinates $(\vartheta,\varphi)$, where $\varphi$ denotes the azimuthal angular coordinate around the axis of symmetry. In particular, the axial Killing vector field $\psi^a$ takes the explicit form $\psi^a = \left(\frac{\partial}{\partial\varphi}\right)^a$, and we denote its corresponding isometry by the action of rotations $\mathfrak{L}$ around the axis of symmetry, i.e., $\mathfrak{L}_\alpha (V,\vartheta,\varphi) := (V,\vartheta,\varphi + \alpha)$ for $\alpha \in \mathbb{R}$. Note, however, that although the notation of $\vartheta$ is inspired by the polar angle from the case where $\mathcal{S}$ is homeomorphic to the $2$-sphere $\mathbb{S}^2$, there also exist other axisymmetric horizon cross sections that are foliated by orbits of $\mathfrak{L}$\footnote{In the cases $\mathcal{S}\cong\mathbb{S}^2$, as well as $\mathcal{S}\cong\mathbb{R}^2$ foliated by concentric circles, there exist values of $\vartheta$ for which the action $\mathfrak{L}$ leaves the points in $\mathcal{S}$ invariant, i.e., the north and south poles for $\mathcal{S}\cong\mathbb{S}^2$, or the origin for $\mathcal{S}\cong\mathbb{R}^2$. These fixed points form a set of measure zero with respect to any Lebesgue-equivalent measure on $\mathcal{S}$, and therefore do not affect the algebraic structure of the QFT described in Section \ref{QFTonKillingHorizons} and Appendix \ref{AppendixMTW}, analogously to the analysis presented in \cite{MTW:2022npa}, where the algebraic decomposability is understood only up to measure-zero subsets in the transverse direction to the null plane.}, for instance a $2$-torus $\mathbb{T}^2 \cong \mathbb{S}^1 \times \mathbb{S}^1$, to give a compact example, or a cylinder $\mathbb{R}\times\mathbb{S}^1$ to give a non-compact one, for both of which $\vartheta$ does not play the role of a polar angle but rather parametrizes $\mathcal{S}$ in a different way.

\subsection{QFT on Bifurcate Killing Horizons}\label{QFTonKillingHorizons}
In this section, we briefly review the structure of a thermal QFT on a bifurcate Killing horizon, according to the results originally presented by Kay and Wald in \cite{KW:1991hadamard}. More specifically, we consider a free massive real scalar field $\Phi$ of mass $m$, i.e., a solution to the Klein-Gordon equation 
\begin{equation}
\label{KleinGordonEquation}
\square_{(\mathcal{M},g)}\Phi = m^2\Phi
\end{equation}

\noindent on a globally hyperbolic axisymmetric stationary spacetime $(\mathcal{M}, g)$ which possesses a bifurcate Killing horizon. Denoting by $\Sigma$ the space of all real, smooth solutions $\Phi$ to \eqref{KleinGordonEquation}, such that $\left. \Phi \right\vert_\mathcal{C}$ and $\left. n^a \nabla_a \Phi \right\vert_\mathcal{C}$ are compactly supported for all Cauchy surfaces $\mathcal{C}\subset\mathcal{M}$ with future-directed unit normal vector $n^a$, we obtain the symplectic space $(\Sigma,\varsigma)$ by introducing the symplectic form
\begin{equation}
\varsigma(\Phi,\Psi) := \int_\mathcal{C} \left( \Phi \, n^a\nabla_a \Psi - \Psi \, n^a \nabla_a \Phi \right) \, d\mathrm{vol}_\mathcal{C}.
\end{equation}

The algebra of observables of the corresponding quantized theory is then given by the unique \cite{BR:1981opalg2} associated Weyl algebra $\mathscr{W}[\Sigma]$ over the symplectic space $(\Sigma, \varsigma)$, which is defined via the Weyl relations
\begin{align}
    W(\Phi)^\dagger &= W (-\Phi) \\
    W(\Phi) W(\Psi) &= e^{\frac{i}{2} \varsigma(\Phi,\Psi)} W(\Phi+\Psi)
\end{align}

\noindent for $W(\Phi),W(\Psi)\in\mathscr{W}[\Sigma]$, $\Phi,\Psi\in\Sigma$  \cite{KW:1991hadamard,Chmielowski:1994pur,DMP:2011uh,FR:2019aqft,KPV:2021ea}. In particular, we consider quasi-free states on $\mathscr{W}[\Sigma]$, i.e., positive linear normalized functionals $\omega:\mathscr{W}[\Sigma]\rightarrow \mathbb{C}$, which are exclusively defined via their corresponding two-point function $\lambda$ by
\begin{equation}
\omega\left(W(\Phi)\right) := e^{-\frac{1}{2}\lambda(\Phi,\Phi)},
\end{equation}

\noindent where
\begin{equation}
\lambda(\Phi,\Psi) = \mu(\Phi,\Psi) + \frac{i}{2}\varsigma(\Phi,\Psi)
\end{equation}

\noindent for some real scalar product $\mu$ on $\Sigma$ that dominates the symplectic form $\varsigma$, cf.\ \cite{KW:1991hadamard}. Moreover, an equivalent symplectic space can be constructed using the causal propagator $E$, i.e., the advanced-minus-retarded fundamental solution to \eqref{KleinGordonEquation} (cf. \cite{BGP:2007weq}). More precisely, one can define a symplectic space by assigning the symplectic form $\sigma(f,g) := \varsigma(Ef,Eg)$ to the quotient space $\mathscr{D}_\mathcal{M}:=C_0^\infty(\mathcal{M})/\ker(E)$ for $f,g\in\mathscr{D}_\mathcal{M}$\footnote{Strictly speaking, the symplectic form $\sigma$ is defined on equivalence classes $[f]\in C_0^\infty(\mathcal M)/\ker(E)=:\mathscr{D}_\mathcal{M}$, where two test functions are equivalent if they only differ by an element in the kernel of the causal propagator $E$, i.e., if they give rise to the same solution of the Klein-Gordon equation \cite{Verch:1997pur}.} \cite{Verch:1997pur}, so that quasi-free states on the associated Weyl algebra $\mathscr{W}[\mathscr{D}_\mathcal{M}]$ are determined via the two-point function
\begin{equation}
\Lambda(f,h) := \lambda(Ef,Eh) = \mu(Ef,Eh) + \tfrac{i}{2}\sigma(f,h).
\end{equation}

In this setting, we specifically consider thermal Hadamard states on $\mathscr{W}[\mathscr{D}_\mathcal{M}]$, i.e., quasi-free states $\omega_H$ that are invariant under the Killing flow $\phi_t$ and satisfy the KMS condition at inverse temperature $\beta > 0$ with respect to $\phi_t$ (cf.  \cite{KW:1991hadamard,Kubo:1957kms,MS:1959kms}) whose two-point function $\Lambda(f,g)$ additionally exhibits the Hadamard singularity structure, or alternatively, fulfills the microlocal spectrum condition (see \cite{FNW:1981hadamard,KW:1991hadamard,Radzikowski:1996ml,FV:2013hadamard,Moretti:2021hadamard,SV:2000pass} for further discussion). In the specific framework considered here, every quasifree KMS state with respect to $\phi_t$ fulfills the Hadamard condition, as shown in \cite{SV:2000pass}. Kay and Wald \cite{KW:1991hadamard} demonstrated that, under these conditions, the restriction of $\Lambda(f,g)$ to the Killing horizon $\mathcal{H}_A \cong \mathbb{R} \times \mathcal{S}$, parameterized by $(V,\vartheta,\varphi)$, reduces to the universal form
\begin{equation}
\label{UniversalScalingLimit}
\Lambda(f,g) = -\frac{1}{\pi} \int \frac{\left. f(V,\vartheta,\varphi) \right\vert_{\mathcal{H}_A} \, \left. g(V',\vartheta,\varphi)\right\vert_{\mathcal{H}_A}}{\left(V - V' - i0^+\right)^2} \, dV\, dV'\, d\mathrm{vol}_\mathcal{S}.
\end{equation}

\noindent In particular, the restriction $\left. \omega_H \right\vert_{\mathcal{H}_A}$ satisfies the KMS condition at inverse temperature $\beta = \frac{2\pi}{\kappa}$ with respect to the projected Killing flow $\left. \phi_t \right\vert_\mathcal{H}$, which precisely agrees with the Hawking temperature \cite{KW:1991hadamard,Hawking:1975rad}. \\

Equivalently, these results can be reformulated in purely algebraic terms following Summers and Verch \cite{SV:1996tomitaki}, using Tomita-Takesaki modular theory \cite{Takesaki:1970tomita,Borchers:2000modth,Summers:2006modth}. Given a von Neumann algebra $\mathscr{A}_\mathcal{R}$ of observables localized in the right wedge $\mathcal{R}$ of a bifurcate Killing horizon, one can find a \emph{symmetry-improving restriction} of $\mathscr{A}_\mathcal{R}$ to $\mathcal{H}_A^\mathcal{R}$, i.e., a subalgebra $\mathscr{N}_\mathcal{R}$ of observables localized on $\mathcal{H}_A^\mathcal{R}$, such that elements in $\mathscr{N}_\mathcal{R}$ transform covariantly with respect to the projected Killing flow $\left. \phi_t \right\vert_\mathcal{H}$ \cite{SV:1996tomitaki}. Considering a KMS state $\omega$ with respect to $\phi_t$ on $\mathscr{A}_\mathcal{R}$, it then follows that its restriction to $\mathscr{N}_\mathcal{R}$ fulfills the KMS condition with respect to $\left. \phi_t \right\vert_\mathcal{H}$ at inverse temperature $\beta = \frac{2\pi}{\kappa}$ \cite{SV:1996tomitaki}, which is, again, the Hawking temperature. Consequently, the horizon theory defined by \eqref{UniversalScalingLimit} is a symmetry-improving restriction of the algebra $\mathscr{A}_\mathcal{R}$ associated to the subalgebra of $\mathscr{W}[\mathscr{D}_\mathcal{M}]$ supported in the region $\mathcal{R}$. In this algebraic formulation, the Hawking temperature thus appears as a consequence of the geometric modular flow of the horizon algebra $\mathscr{N}_\mathcal{R}$ associated with the theory on $\mathcal{H}_A^\mathcal{R}$. \\

More precisely, denoting by $\Omega_\omega\in\mathscr{H}_\omega$ the cyclic vector of the GNS representation corresponding to the KMS state $\omega$ on $\mathscr{N}_\mathcal{R}$, and recalling that the KMS condition can be formulated in terms of the modular flow \cite{Summers:2006modth}, it follows that $\Omega_\omega$ is also separating for $\mathscr{N}_\mathcal{R}$, and that the modular flow corresponds to the projected Killing flow $\left.\phi_{\beta_\mathrm{KMS} t} \right\vert_\mathcal{H}$, which is equivalent to the modular operator $\Delta_\mathcal{R}$ with respect to $(\mathscr{N}_\mathcal{R}, \Omega_\omega)$ acting as affine dilations along $\mathcal{H}_A$ \cite{SV:1996tomitaki}, i.e.,
\begin{equation}
\Delta_\mathcal{R}^{it} = \mathtt{U}_{2\pi t}.
\end{equation}

\noindent Here, $\mathtt{U}_t$ denotes the unitary action $\mathtt{U}_t W(f) \mathtt{U}_t^\dagger = W(\mathfrak{D}_t^\ast f)$, where $\mathfrak{D}_t^\ast f(V,p) = f(e^{-t}\,V,p)$ denotes the pullback of $f\in\mathscr{D}_\mathcal{M}$ by $\mathfrak{D}$. Similarly, we may also consider the action of affine translations $\mathfrak{T}_a (V,p) = (V+a,p)$, for which the corresponding unitaries $\mathtt{V}_a$ act via the pullback $\mathfrak{T}_a^\ast f(V,p) = f(V-a,p)$, i.e. $\mathtt{V}_a W(f) \mathtt{V}_a^\dagger = W(\mathfrak{T}_a^\ast f)$. Together with the algebra $\mathscr{N}_\mathcal{R}$ and the cyclic vector $\Omega_\omega$, these unitaries $\mathtt{V}_a$ form a Borchers triple $(\mathscr{N}_\mathcal{R}, \mathtt{V}_a, \Omega_\omega)$, see \cite{Borchers:1992cpt}. \\

A further structural aspect of the horizon theory on $\mathcal{H}_A^\mathcal{R}$ becomes apparent when comparing it with the recent analysis of null-surface algebras by Morinelli, Tanimoto, and Wegener \cite{MTW:2022npa}. In that work, free scalar field observables localized on a null plane are shown to admit a transverse decomposition in the following sense. Given a suitable representation, the null plane von Neumann algebra acts on a Hilbert space that decomposes as a direct integral over the transverse parameter, and its elements are decomposable with respect to this direct-integral structure \cite{MTW:2022npa}. We provide further details on the decomposability of null-surface algebras in Appendix \ref{AppendixMTW}. In close analogy with this picture, the horizon algebra $\mathscr{N}_\mathcal{R}$ can be understood as admitting a transverse decomposition along the azimuthal coordinate $\varphi\in(-\pi,\pi)$ (for some fixed $\vartheta$), with each value of $\varphi$ labelling a null geodesic in $\mathcal{H}_A^\mathcal{R}$, as sketched in Figure \ref{FigureNullCone}. In this sense, $\mathscr{N}_\mathcal{R}$ may be interpreted as a direct integral $\int_{(-\pi,\pi)}^\oplus \mathscr{N}_\varphi \, d\varphi$ \cite{Takesaki:1979oa,MTW:2022npa} over identical von Neumann algebras $\mathscr{N}_\varphi\cong\mathscr{N}_0$ along the $\varphi$-direction, each describing a chiral conformal field theory (CFT) along the corresponding null ray (cf. \cite{Hollands:2020relent}). \\

\begin{figure}[h!]
\centering
\includegraphics[width=0.71\textwidth]{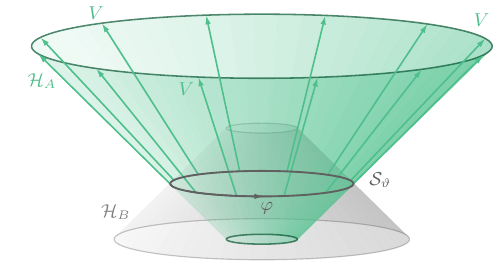}
\caption{\justifying \small Schematic illustration of the future Killing horizon $\mathcal{H}_A$ (green cone), and a congruence of null geodesics along the horizon (green rays), affinely parametrized by the lightlike coordinate $V$. Fixing the coordinate $\vartheta$, the corresponding circular cross section $\mathcal{S}_\vartheta$ serves as a transverse section of the horizon, while the azimuthal coordinate $\varphi\in(-\pi,\pi)$ parametrizes the individual null rays within the congruence. In analogy with the null-plane analysis of \cite{MTW:2022npa}, the horizon theory on $\mathcal{H}_A$ admits a natural interpretation in terms of a decomposition into identical chiral CFTs propagating along the affine parameter $V$, labelled by $\varphi$.}
\label{FigureNullCone}
\end{figure}

Finally, we point out that analogous horizon theories also arise in more general situations. In particular, the scaling-limit construction in \cite{KPV:2021ea} yields a near-horizon theory in dynamical spherically symmetric spacetimes which exhibits the same structural properties as a symmetry improving restriction and admits a corresponding direct integral decomposition, as described above. Moreover, perturbative results for $\Phi^3$-theory in \cite{CMP:2014phi3} indicate that these features remain compatible with interacting models at one-loop level.

\subsection{Rieffel Deformations and Noncommutative Spacetimes}
Rieffel deformations \cite{Rieffel:1993defq} provide an algebraic deformation procedure in which the pointwise product is replaced by an associative but noncommutative one. It is constructed using a strongly continuous group action $\alpha$ of $\mathbb{R}^n$ ($n\in\mathbb{N}$) on a $C^\ast$-algebra $\mathscr{A}$, in order to define a deformed product $\times_\Theta$ on a dense subalgebra $\mathscr{A}^\infty$ of elements that are bounded and smooth with respect to $\alpha$. More specifically, $\times_\Theta$ is characterized by a real skew-symmetric matrix $\Theta \in \mathbb{R}^{n \times n}$ and is explicitly defined as the oscillatory integral
\begin{equation}
\label{RieffelProduct}
(f \times_\Theta g)(x) := \frac{1}{(2\pi)^n} \lim_{\varepsilon\rightarrow0} \int_{\mathbb{R}^n \times \mathbb{R}^n} \chi(\varepsilon u , \varepsilon v)\,e^{-i x \cdot y} \,\alpha_{\Theta u}(f(x))\,\alpha_{v}(g(x)) \, d^nu \, d^nv,
\end{equation}

\noindent which converges for all Hörmander symbols $f,g\in S^m_\rho(\mathbb{R}^n)$ \cite{Hoermander:1989osc}, and arbitrary Schwartz functions $\chi\in\mathscr{S}(\mathbb{R}^n\times\mathbb{R}^n)$ which fulfill $\chi(0,0)=1$ \cite{Rieffel:1993defq, BLS:2011wcrd}. \\

Generally, the Rieffel product $\times_\Theta$ is associative but not commutative, unital, i.e., $1\times_\Theta f = f\times_\Theta 1 = f$, and it reduces to the usual pointwise product for $\Theta \rightarrow 0$, i.e., $f \times_0 g = fg$ for all $f,g\in S^m_\rho(\mathbb{R}^n)$ \cite{Much:2021mdga}. In particular, is is shown in \cite{FM:2020defds} that the Rieffel product is equivalent to the Groenewold-Moyal-Weyl product \cite{Groenewold:1946prod,Moyal:1949prod}, i.e.,
\begin{equation}
\label{GroenewoldMoyalWeyl}
(f\times_\Theta g)(x) = \left. e^{i\Theta^{ab}\frac{\partial}{\partial y^a}\frac{\partial}{\partial z^b}} f(y) \, g(z) \right\vert_{y=z=x}.
\end{equation}

When applied to QFT on ($n$-dimensional) Minkowski spacetime, this deformation procedure offers a systematic way to effectively construct quantum field theoretical models on a noncommutative background. Physically, the deformation is motivated by the idea that spacetime itself acquires a quantum description at the smallest scales, i.e., around the Planck length \cite{DFR:1994lett,DFR:1995ncs}, leading to uncertainty relations between spacetime coordinates of the form
\begin{equation}
[x^a, x^b] = 2i \Theta^{ab},
\end{equation}

\noindent in analogy to the canonical commutation relations between position and momentum operators in quantum mechanics. In fact, a straightforward computation shows that the commutator with respect to the Rieffel product of coordinate components reproduces the above commutation relation exactly, i.e.,
\begin{equation}
x^a \times_\Theta x^b - x^{b} \times_\Theta x^a = 2i\Theta^{ab}.
\end{equation}

More precisely, Rieffel deformations are employed in the algebraic formulation of QFT to construct theories described by deformed algebras that respect fundamental physical principles, such as Poincaré covariance or (wedge-) locality \cite{BLS:2011wcrd}. Furthermore, a concrete representation of Rieffel deformed algebras is provided by \emph{warped convolutions}, which are compatible with the Rieffel product $\times_\Theta$ in the sense that $\pi_\Theta(f)\pi_\Theta(g) = \pi_\Theta(f\times_\Theta g)$ \cite{BLS:2011wcrd}. This approach rigorously generalizes the explicit construction of quantum field on noncommutative Minkowski spacetime originally introduced by Grosse and Lechner in \cite{GL:2007fields}. \\

Despite their success in Minkowski spacetime, generalizing Rieffel deformations and their application to QFT to curved backgrounds poses significant challenges, mainly due to the absence of global symmetries, which are essential for the original construction of the Rieffel product $\times_\Theta$, which relies heavily on the existence of a strongly continuous action of the translation group. Nonetheless, notable progress has been made in extending such deformations to more general geometric settings. More specifically, manifolds that possess a Poisson structure allow for the construction of deformed products \cite{Fedosov:1994star, Kontsevich:2003def, Waldmann:2007poisson}. \\

Most notably, a generalized Rieffel deformation has been introduced in \cite{Much:2021mdga, Much:2025genrd}, based on geodesic transport as a generalization of translations. This construction yields an associative product at least up to second order in the deformation parameter, which is well-defined for QFT on globally hyperbolic spacetimes. Within this framework, it has been shown that the Hadamard condition is preserved under the generalized deformation, ensuring that the singularity structure of physically admissible states remains intact \cite{Much:2021mdga}. Motivated by these developments, the present work aims to construct a Rieffel-type deformation defined via the action of the symmetry group on the Killing horizon of a general axisymmetric stationary spacetime. The resulting star product can be applied to QFT on Killing horizons, enabling the derivation of closed-form expressions for various physically interesting quantities.

\section{A Novel Rieffel-type Product on Killing Horizons}\label{ConstructionAxisymmetricProduct}
While the original Rieffel product $\times_\Theta$ is defined via the isometries of spatiotemporal translations in Minkowski spacetime, or $\mathbb{R}^n$, respectively, we aim to construct a similar deformed product $\star_\Theta$ generated by the isometries of Killing horizons in stationary and axisymmetric spacetimes, i.e., dilations $\mathfrak{D}$ along the affine parameter, and rotations $\mathfrak{L}$ around the axis of symmetry, which commute, analogously to their corresponding generating Killing vector fields \cite{Carter:1970kc}.\\

First of all, we require a suitable definition of a function space, namely compactly supported functions on $\mathcal{H}_A^\mathcal{R}$, where angular coordinates and, in particular, the azimuthal angle $\varphi$, will be considered as a variable in a \emph{single copy} of the coordinate chart on the submanifold $\mathcal{S}$. Note that we restrict ourselves to the right horizon $\mathcal{H}_A^\mathcal{R}$ from the outset, corresponding to the open half-line $V>0$ rather than the entire real line $V\in\mathbb{R}$. Not only is this restriction necessary to avoid the fixed point $V=0$ of the dilation action $c\mapsto e^cV$ (see Definition \ref{OriginalDefinitionDeformedProduct} below), but $\mathcal{H}_A^\mathcal{R}$ is also the natural domain for the algebraic framework constructed in Section \ref{SecDeformedRelativeEntropy}, in which the "right" Weyl algebra $\mathscr{W}_\Theta^\mathcal{R}$, the associated KMS state, and the relative entropy are formulated (cf. the treatment of the right wedge in \cite{SV:1996tomitaki,KPV:2021ea}). The corresponding analysis on the left horizon $\mathcal{H}_A^\mathcal{L}$ ($V<0$) is entirely analogous up to signature changes.

\begin{definition}\label{DefinitionFunctionSpace}
We define the function space $\mathscr{D}_{\mathcal{H}_A}$ as the space of smooth compactly supported functions on the right Killing horizon $\mathcal{H}_A^\mathcal{R} \cong \mathbb{R}_{+}\times\mathcal{S}$, i.e., on the open half-line $V>0$, where compact angular support is understood as a strictly compact subset $\mathcal{C}$ of the range of a given coordinate chart on $\mathcal{S}$, i.e., $\mathcal{C} \subset \mathrm{Range}(\vartheta,\varphi) \subseteq \mathbb{R}^2$ compact.
\end{definition}

\noindent \textbf{Remark.} In this work, we adopt the coordinate range $(-\pi, \pi)$ for the azimuthal angle $\varphi$ on the circle. In principle, other coordinate choices could also be considered, which would lead to different definitions of deformed products.

%We like to point out, however, that alternative choices of coordinates are equally valid, e.g., a coordinate chart in which $\varphi$ takes values in $\mathbb{R}$, as obtained via stereographic projection from the circle (excluding a single pole) onto the real line.

\begin{definition}\label{OriginalDefinitionDeformedProduct}
Let $z = {c \choose \gamma} \in \mathbb{R}^2$ and denote by $\pi_z$ the corresponding action
\begin{align}
\pi_z (f)(V,\vartheta,\varphi)  := f\left(\mathfrak{D}_c \mathfrak{L}_\gamma (V,\vartheta,\varphi)\right) = f\left(\mathfrak{L}_\gamma\mathfrak{D}_c (V,\vartheta,\varphi)\right) = f\left( e^{c} V, \vartheta, \varphi + \gamma \right)
\end{align}

\noindent Let furthermore $\Theta$ be an antisymmetric matrix with respect to the product
\begin{equation}
    (x\vert y) := \left( {a \choose \alpha} \Biggl\vert {b \choose \beta} \right) := ab + \alpha\beta
\end{equation}

\noindent on $\mathbb{R}^2$, i.e.,
\begin{equation}
    \Theta = \left( \begin{array}{cc}  0 & \theta \\ - \theta & 0 \end{array} 
    \right)
\end{equation}
with $\theta \in \mathbb{R}$ chosen arbitrary but fixed. Obviously, $\Theta$ is invertible if and only if $\theta \ne 0$. In the following, we generally consider 
$\theta \ne 0$ unless explicitly specified otherwise.\\

Then, we define the product $\star_\Theta$ for functions $f,g\in \mathscr{D}_{\mathcal{H}_A}$ by the oscillatory integral
\begin{equation}
\label{DilationProduct}
(f \star_\Theta g)(V,\vartheta,\varphi) := \lim_{\varepsilon\rightarrow0}  \frac{1}{(2\pi)^2} \int_{\mathbb{R}^2\times\mathbb{R}^2} \chi(\varepsilon x, \varepsilon y) \, e^{-i(x \vert y)} \mbox{\large{$($}}\pi_{\Theta x} (f) \pi_y (g)\mbox{\large{$)$}}(V,\vartheta,\varphi) \, dx \, dy,
\end{equation}

\noindent where $\chi\in\mathscr{S}(\mathbb{R}^2\times\mathbb{R}^2)$ is a Schwartz function (cf. \cite{Rieffel:1993defq,Hoermander:1989osc}) with the property that $\chi(0,0)=1$.
\end{definition}

%\noindent \textbf{Remark.} One may alternatively define a similar Rieffel product with respect to affine translations $\mathfrak{P}$, which is another isometry group of $\mathcal{H}_A$ \cite{SV:1996tomitaki}, and which also commutes with $\mathfrak{L}$. However, $\mathfrak{P}$ does not commute with the affine translations $\mathfrak{D}$, which would lead to severe complications regarding latter considerations in this work. 

\begin{theorem}\label{TheoremExistenceProperties}
The oscillatory integral \eqref{DilationProduct} converges absolutely for all $f,g\in\mathscr{D}_{\mathcal{H}_A}$.
\end{theorem}

\noindent \textbf{Proof.} Let us consider the support properties of $\pi_z(f)$ in $\mathbb{R}^2\ni z = {c \choose \gamma}$ for some $f\in\mathscr{D}_{\mathcal{H}_A}$ and fixed values of $(V,\vartheta,\varphi)\in \mathbb{R}_+\times\mathcal{S}$. By Definition \ref{DefinitionFunctionSpace}, the support of $f$ lies in the compact interval $[V_\mathrm{min},V_\mathrm{max}]\times \mathcal{C} \subset \mathbb{R}_+ \times \mathcal{S}$ with $0<V_\mathrm{min} \leq V_\mathrm{max} < \infty$ and $[\vartheta_\mathrm{min},\vartheta_\mathrm{max}]\times[\varphi_\mathrm{min},\varphi_\mathrm{max}]\subset\mathcal{I}_\vartheta\times(-\pi,\pi)$ compact, where $\mathcal{I}_\vartheta$ denotes the coordinate range of $\vartheta$. Consequently, $\pi_z(f)$ vanishes for $c\notin \left[\ln\left(V_\mathrm{min}/V\right),\ln\left( V_\mathrm{max}/V\right)  \right]$, and similarly for $\gamma\notin\left[ \varphi_\mathrm{min}-\varphi , \varphi_\mathrm{max}-\varphi \right]$. Accordingly, we find that for any $f,g\in\mathscr{D}_{\mathcal{H}_A}$ and $y={b\choose \beta}$, the functions $\pi_{\Theta z}(f)$ and $\pi_y(g)$ are, for fixed $(V,\vartheta,\varphi)\in\mathbb{R}_+\times\mathcal{S}$ and nonzero deformation parameter $\theta$, supported within the compact intervals 
\begin{align}
\gamma&\in\left[ -\frac{1}{\theta}\ln\left(\frac{V^f_\mathrm{max}}{V} \right) , -\frac{1}{\theta}\ln\left(\frac{V^f_\mathrm{min}}{V}\right) \right], \quad
c \in\left[ \frac{\varphi^f_\mathrm{min}-\varphi}{\theta} , \frac{\varphi^f_\mathrm{max}-\varphi}{\theta} \right],
\end{align}

\noindent and
\begin{align}
b&\in\left[ \ln\left(\frac{V^g_\mathrm{min}}{V} \right) , \ln\left(\frac{V^g_\mathrm{max}}{V}\right) \right], \quad
\beta \in\left[ \varphi^g_\mathrm{min}-\varphi , \varphi^g_\mathrm{max}-\varphi \right],
\end{align}

\noindent respectively. Hence, both $\mathrm{supp}_{z}\left(\pi_{\Theta z}(f)\right)\subset\mathbb{R}^2$ and $\mathrm{supp}_{y}\left(\pi_{y}(g)\right)\subset\mathbb{R}^2$ remain compact, which directly yields that the oscillatory integral \eqref{DilationProduct} not only exists, but also that
\begin{equation}
(f \star_\Theta g)(V,\vartheta,\varphi) = \frac{1}{(2\pi)^2} \int_{\mathbb{R}^2\times\mathbb{R}^2} \left\vert \pi_{\Theta z} (f)(V,\vartheta,\varphi) \right\vert \left\vert\pi_y (g)(V,\vartheta,\varphi)\right\vert \, dz \, dy < \infty.
\end{equation}

$\qedsymbol$

\begin{proposition}\label{PropositionExpansion}
The deformed product \eqref{DilationProduct} is formally equivalent to
\begin{equation}
\label{FormalExpansionProduct}
(f\star_\Theta g)(V,\vartheta,\varphi) = \left. e^{i\theta\left( X\partial_X \partial_{\eta} - Y\partial_{Y} \partial_\xi \right)} f(X,\vartheta,\xi) \, g(Y,\vartheta,\eta) \right\vert_{\substack{X=Y=V \\ \xi=\eta=\varphi}},
\end{equation}

\noindent i.e., the expansion
\begin{equation}
\label{FiniteExpansionProduct}
(f\star_\Theta g)(V,\vartheta,\varphi) = \sum_{k=0}^N \frac{(i\theta)^k}{k!} \left(X\partial_X \partial_\eta -Y\partial_Y \partial_\xi\right)^k f(X,\xi)\, g(Y,\eta) \bigg\vert_{\substack{X=Y=V \\ \xi=\eta=\varphi}} + \mathcal{O}\left(\theta^{N+1}\right)
\end{equation}

\noindent holds for any finite order $N\in\mathbb{N}$ in the deformation parameter $\theta$.
\end{proposition}

\noindent The proof for this statement is given in Appendix \ref{ProofPropositionExpansion}. \\

\noindent \textbf{Remark.} For analytic functions $f$ and $g$ like Hermite polynomials, or more generally, functions in a suitable Gelfand-Shilov space (cf. \cite{GS:1968og,Soloviev:2007gs}), the series expansion \eqref{FormalExpansionProduct} may converge absolutely, and then exactly equal the oscillatory integral \eqref{DilationProduct}. However, for non-analytic functions such as functions in $\mathscr{D}_{\mathcal{H}_A}$, for which the oscillatory integral \eqref{DilationProduct} converges absolutely, or more generally for functions in other $C_0^\infty$ spaces, the formal power series \eqref{FiniteExpansionProduct} in general fails to converge for $N\rightarrow\infty$. Consequently, one may only consider truncations of the series at arbitrary but finite order in the deformation parameter $\theta$. Nonetheless, this is fully sufficient for the purposes of the present work, where the formal expansion is used to extract noncommutative correction terms for physical quantities in QFT. \\

A closely related point concerns the localization properties of the deformed product. Even if $f,g\in\mathscr{D}_{\mathcal{H}_A}$, the deformed product $f\star_\Theta g$ will, in general, no longer belong to the same algebra of functions. Instead, it lies in some larger, extended algebra whose support generally exceeds that of the original test functions, both in $V$ and $\varphi$. In particular, the $\varphi$-support generally exceeds the interval $(-\pi,\pi)$. Nevertheless, Proposition \ref{PropositionExpansion} yields that the product $f\star_\Theta g$ admits the decomposition
\begin{equation}
f \star_\Theta g =: F_\Theta := F_N + F_\infty
\end{equation}

\noindent where the localized contribution $F_N$ is given by the finite sum \eqref{FiniteExpansionProduct}, while the remainder term $F_\infty$ includes all higher-order contributions in the deformation parameter $\theta$. By construction, $F_N$ only involves finitely many derivatives of $f$ and $g$ and therefore remains an element of $\mathscr{D}_{\mathcal{H}A}$, which we will establish in the following Lemma, whereas the remainder $F_\infty$ generally exhibits a delocalization beyond the test functions' original support, as illustrated in Figure \ref{FigureSupportExtension}.

\begin{lemma}\label{LemmaSupportPreservation}
If $f,g \in \mathscr{D}_{\mathcal{H}_A}$, then it holds for
\begin{equation}
F_N = \sum_{k=0}^N \frac{(i\theta)^k}{k!} \left(X\partial_X \partial_\eta -Y\partial_Y \partial_\xi\right)^k f(X,\xi)\, g(Y,\eta) \bigg\vert_{\substack{X=Y=V \\ \xi=\eta=\varphi}}
\end{equation}

\noindent that $F_N\in \mathscr{D}_{\mathcal{H}_A}$ with $\mathrm{supp}(F_N) \subseteq \mathrm{supp}(f)\cap\mathrm{supp}(g)$ for all $N\in\mathbb{N}$.
\end{lemma}

\noindent \textbf{Proof.} This property follows immediately from the fact that $\mathrm{supp}(\frac{d^n}{dx^n} \mathfrak{f})\subseteq \mathrm{supp}(\mathfrak{f})$ for all $\mathfrak{f}\in C_0^\infty(\mathbb{R})$, which can straightforwardly be generalized to higher dimensions and different partial derivatives. Considering the expansion \eqref{FiniteExpansionProduct}, we consequently observe that each order is compactly supported within a subset of $\mathrm{supp}(f)\cap\mathrm{supp}(g)$. $\qedsymbol$ \\

From a physical perspective, this ensures that, up to any finite order in $\theta$, the deformation preserves the original localization structure of observables, both in the null variable $V$ and the angular variable $\varphi$, and allows the interval $(-\pi,\pi)$ to be consistently identified with the circle, also after the deformation. Only the higher-order contributions from $F_\infty$ require a description in terms of an extended algebra, which may be interpreted as corresponding to an extended theory with localization beyond the original domain, particularly with respect to $\varphi$. Since the present work is primarily concerned with deriving lower-order noncommutative corrections to physical observables, the higher-order contributions contained in $F_\infty$ can, for our purposes, be regarded as negligible, in the sense that $F_\infty=\mathcal{O}(\theta^{N+1}) \ll F_N$.

\begin{figure}[h!]
\centering
\includegraphics[width=0.84\textwidth]{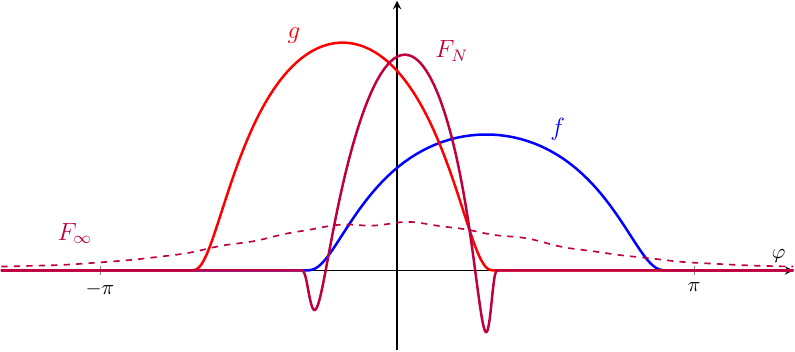}
\caption{\justifying \small Qualitative sketch illustrating the decomposition $f\star_\Theta g =: F_\Theta := F_N + F_\infty$. The functions $f,g\in\mathscr{D}_{\mathcal{H}_A}$ (drawn in blue and red) have compact $\varphi$-support contained within $(-\pi,\pi)$. The finite-order contribution $F_N$, defined by \eqref{FiniteExpansionProduct}, is supported in $\mathrm{supp}(F_N) \subseteq \mathrm{supp}(f)\cap\mathrm{supp}(g)$, as depicted by the continuous purple curve. The higher-order remainder $F_\infty$ is shown as a dashed curve, indicating that its support generally extends beyond $(-\pi,\pi)$. This illustrates that, while finite truncations preserve the localization within $\mathscr{D}_{\mathcal{H}_A}$, the full deformed product generally belongs to an extended algebra. A similar picture holds for the $V$-support of $f$ and $g$.}
\label{FigureSupportExtension}
\end{figure}

\begin{corollary}
The first-order expansion of the deformed product \eqref{DilationProduct} in the deformation parameter $\theta$ is given in terms of the Poisson bracket
\begin{equation}
\left\lbrace f,g \right\rbrace_\Theta := \left( V\partial_V f(V,\vartheta,\varphi) \partial_\varphi g(V,\vartheta,\varphi) - V\partial_V g(V,\vartheta,\varphi) \partial_\varphi f(V,\vartheta,\varphi) \right)
\end{equation}

\noindent for $f,g\in\mathscr{D}_{\mathcal{H}_A}$ by the perturbative expression
\begin{equation}
(f \star_\Theta g)(V,\vartheta,\varphi) = (fg)(V,\vartheta,\varphi) + i\theta\left\lbrace f,g \right\rbrace_\Theta(V,\vartheta,\varphi) + \mathcal{O}(\theta^2).
\label{FirstOrderProduct}
\end{equation}
\end{corollary}

This result agrees with well-established core properties of star products in Moyal-Rieffel-type deformations, where the first-order term of any strict or formal star product is governed by the Poisson bracket associated with the underlying (sub-)manifold, provided it admits a Poisson structure \cite{Fedosov:1994star,Weinstein:1994,Kontsevich:2003def,dWL:1985def,Waldmann:2007poisson}. In particular, the first order in the deformation parameter $\theta$ of the commutator $[f,g]_{\star_\Theta}$ with respect to the deformed product $\star_\Theta$ recovers this Poisson bracket exactly. Geometrically, the Poisson bracket arises from the intrinsic Poisson structure of the background (sub-)manifold and is determined by the Poisson tensor characterizing the deformation, as originally described in \cite{Fedosov:1994star}, and more recently in \cite{Much:2021mdga}.

\begin{corollary}
The second-order expansion of the deformed product \eqref{DilationProduct} in the deformation parameter $\theta$ is given by
\begin{align}
\label{SecondOrderProduct}
(f \star_\Theta g)(V,\vartheta,\varphi) &= fg + i\theta\left\lbrace f,g \right\rbrace_\Theta \\
&\hspace{0.5cm}- \frac{\theta^2}{2} \big( V^2 (\partial_V^2 f)(\partial_\varphi^2 g) + V(\partial_V f)(\partial_\varphi^2 g) + V^2 (\partial_V^2 g)(\partial_\varphi^2 f) + V(\partial_V g)(\partial_\varphi^2 f) \nonumber \\
&\hspace{1.67cm} - 2V^2 (\partial_V\partial_\varphi f)(\partial_V \partial_\varphi g) \big) + \mathcal{O}(\theta^3) \nonumber
\end{align}

\noindent for $f,g\in\mathscr{D}_{\mathcal{H}_A}$.
\end{corollary}

\begin{theorem}\label{TheoremUnitality}
The deformed product \eqref{DilationProduct} is unital, i.e.,
\begin{equation}
(f\star_\Theta 1)(V,\vartheta,\varphi) = (1\star_\Theta f)(V,\vartheta,\varphi)=f(V,\vartheta,\varphi)
\end{equation}

\noindent for all $f\in\mathscr{D}_{\mathcal{H}_A}$.
\end{theorem}

\noindent \textbf{Proof.} Let us consider the oscillatory integral
\begin{align}
\nonumber
(f \star_\Theta 1)(V,\vartheta,\varphi) &= \frac{1}{(2\pi)^2}\int_{\mathbb{R}^2\times\mathbb{R}^2} e^{-i(ab+\alpha\beta)} \pi_{\Theta x}(f) \, \pi_y(1) \, da\,db\;d\alpha\,d\beta \\
&= \frac{1}{(2\pi)^2} \int_{\mathbb{R}^2\times\mathbb{R}^2} e^{-iab} e^{-i\alpha\beta} f(e^{-\theta\alpha}V,\vartheta,\varphi+\theta a) \, da\, db\; d\alpha\, d\beta,
\label{UnitalIntegral}
\end{align}

\noindent where we have used that $1$ is invariant under the action $\pi_{y}$. Given that $b$ and $\beta$ only appear in the oscillatory term in expression \eqref{UnitalIntegral}, we use the Fourier identity
\begin{equation}
\int_\mathbb{R} e^{-iXY} dY = 2\pi \delta(X),
\end{equation}

\noindent so that we obtain
\begin{equation}
(f \star_\Theta 1)(V,\vartheta,\varphi) = \int_{\mathbb{R}\times\mathbb{R}} f(e^{-\theta\alpha}V,\vartheta,\varphi+\theta a) \, \delta(a)\,\delta(\alpha)\, db\,d\beta = f(V,\vartheta,\varphi).
\end{equation}

\noindent An analogous computation also leads to $(1\star_\Theta f)=f$. $\qedsymbol$ \\

\noindent \textbf{Remark.}
We like to point out that, strictly speaking, the constant function $1\in C^\infty(\mathbb{R}_+\times\mathcal{S})$ does not belong to the function space $\mathscr{D}_{\mathcal{H}_A}$, since any function with compact support contained in the interval $(-\pi,\pi)$ necessarily vanishes at the boundary of that interval and is nonzero only in the interior, and hence cannot be constant. Nevertheless, the star product $\star_\Theta$ remains well-defined if one or both of the factors are constant. This is a consequence of the fact that constant functions are invariant under the generating group actions of affine dilations and angular shifts and may be viewed as trivial examples of Hörmander symbols of order zero (cf. \cite{Hoermander:1989osc,Rieffel:1993defq}) as all of their derivatives vanish identically. Consequently, the oscillatory integral defining $\star_\Theta$ reduces to the undeformed pointwise product in this situation, as outlined in the proof above. More generally, extending the domain of $\star_\Theta$ to include constant functions and thereby ensuring unitality, with all nontrivial deformation terms vanishing whenever a constant function is involved, is a standard and well-established feature of star product deformations \cite{Kontsevich:2003def,Waldmann:2007poisson,Much:2021mdga}. From a physical perspective, this guarantees that all possible zero mode contributions, such as functions that are constant in the azimuthal angle, are naturally included in the framework without the need for any additional prescriptions.

\begin{theorem}\label{TheoremCommutativeLimit}
In the commutative case, i.e., for $\theta = 0$ or equivalently for $\Theta=\mathbb{O}^{2\times2}$, the deformed product \eqref{DilationProduct} consistently coincides with the pointwise product, i.e.,
\begin{equation}
(f\times_0\, g)(V,\vartheta,\varphi) = (fg)(V,\vartheta,\varphi)
\end{equation}

\noindent for all $f,g\in\mathscr{D}_{\mathcal{H}_A}$.
\end{theorem}

\noindent \textbf{Proof.} Let us consider the oscillatory integral
\begin{align}
\nonumber
(f\times_0 g)(V,\vartheta,\varphi) &= \frac{1}{(2\pi)^2}\int_{\mathbb{R}^2\times\mathbb{R}^2} e^{-i(ab+\alpha\beta)} \pi_{0}(f) \, \pi_y(g) \, da\,db\;d\alpha\,d\beta \\
&= \frac{1}{(2\pi)^2} \int_{\mathbb{R}^2\times\mathbb{R}^2} e^{-iab} e^{-i\alpha\beta} f(V,\vartheta,\varphi) \, g(e^bV,\vartheta,\varphi+\beta)\, da\, db\; d\alpha\, d\beta.
\end{align}

\noindent Since the argument of $f$ is independent of $a$ and $\alpha$ in this expression, we use the same arguments as in the previous proof to reduce the integral to
\begin{align}
(f\times_0 g)(V,\vartheta,\varphi) = f(V,\vartheta,\varphi) \underbrace{\int_{\mathbb{R}\times\mathbb{R}}  \, g(e^bV,\vartheta,\varphi+\beta) \, \delta(b)\, \delta(\beta) \,db\,d\beta}_{=g(V,\vartheta,\varphi)}.
\end{align}

$\qedsymbol$

\begin{theorem}\label{TheoremAssociativity}
The deformed product \eqref{DilationProduct} is associative in the sense of formal power series, yielding
\begin{equation}
\big((f\star_\Theta g)\star_\Theta h\big)(V,\vartheta,\varphi) = \big(f\star_\Theta (g \star_\Theta h)\big)(V,\vartheta,\varphi) + \mathcal{O}\left(\theta^{N+1}\right)
\end{equation}

\noindent up to any arbitrary finite order $N\in\mathbb{N}$ in the deformation parameter $\theta$, for all $f,g,h\in\mathscr{D}_{\mathcal{H}_A}$.
\end{theorem}

\noindent The proof for this statement is given in Appendix \ref{ProofAssociativity}. \\

The formal associativity of the deformed product $\star_\Theta$ is a direct consequence of the fact that it is constructed via two commuting group actions \cite{Gutt:2018ga} of the isometry group of the underlying spacetime (sub-)manifold, acting on the function space $\mathscr{D}_{\mathcal{H}_A}$, which in our case are given by affine dilations $\mathfrak{D}$ and angular shifts $\mathfrak{L}$. \\

Altogether, the deformed product $\star_\Theta$ on $\mathscr{D}_{\mathcal{H}_A}$ satisfies the standard criteria for a formal star product on Poisson manifolds, see, e.g., \cite{Fedosov:1994star,Kontsevich:2003def,Waldmann:2007poisson,Much:2021mdga}. It is (formally) associative, unital, reduces to the commutative pointwise product for $\theta = 0$, and its first-order expansion reproduces the underlying Poisson bracket. In particular, its construction is geometrically motivated, being generated by the commuting actions of the isometry group of the background spacetime, which ensures many of the aforementioned properties and renders the product $\star_\Theta$ a promising candidate for a suitable algebraic deformation for QFT on axisymmetric Killing horizons.

\section{Deformed Symplectic Spaces}\label{SecDeformedQFT}
In standard Moyal-Rieffel deformations, noncommutativity is introduced by deforming the Poisson structure underlying the phase space of classical observables. Whenever the corresponding Poisson bivector is invertible, this is equivalent to deforming the associated symplectic form \cite{Much:2021mdga}. In (scalar) QFT, however, no such  phase space is formulated \emph{a priori}. Instead, the canonical equal-time commutation relations are intrinsically encoded in the symplectic form on the infinite-dimensional space of field configurations and conjugate momenta. \\

Accordingly, we propose to treat deformations of this classical symplectic structure as the field-theoretic analogue to phase space deformations in quantum mechanics. Consequently, rather than deforming the   Weyl algebra on the horizon, we construct the Weyl algebra associated with the deformed symplectic form on the underlying space of field configurations. This strategy conceptually aligns with the formal deformation quantization approach for free fields in perturbative AQFT in \cite[Section 3.5]{FR:2015paqft}. \\

Recalling the theoretical setup introduced in Section \ref{QFTonKillingHorizons}, the symmetry improved restriction on the future right Killing horizon $\mathcal{H}_A^\mathcal{R} \cong \mathbb{R}_+ \times \mathcal{S}$, emanating from the bifurcation surface $\mathcal{S}$, is described by the vector space $\mathscr{D}_{\mathcal{H}_A}$ of smooth, compactly supported test functions on $\mathcal{H}_A^\mathcal{R}$ (cf.\ Definition \ref{DefinitionFunctionSpace}), which is equipped with the symplectic form
\begin{equation}
\label{UndeformedSymplecticForm}
\sigma(f,g) := \int_{\mathcal{H}_A^\mathcal{R}} \big( (\partial_V f)\,g - (\partial_V g)\,f\big)\, dV\, d\mathrm{vol}_\mathcal{S}.
\end{equation}

\noindent The resulting classical symplectic space $(\mathscr{D}_{\mathcal{H}_A}, \sigma)$ is then quantized according to the previously introduced procedure, i.e., by considering the associated Weyl algebra of observables $\mathscr{W}$ on $\mathcal{H}_A$ \cite{KPV:2021ea}. \\

To effectively implement a noncommutative background geometry in the sense of \cite{DFR:1995ncs,BDMP:2015qst} into the horizon theory, we deform the function space $\mathscr{D}_{\mathcal{H}_A}$ with respect to the star product \eqref{DilationProduct}. As discussed above, this product does not, in general, map $\mathscr{D}_{\mathcal{H}_A}$ into itself but instead takes values in a larger, delocalized function space. Guided by the conceptual framework outlined above, namely the deformation of a QFT through a deformation of the underlying symplectic structure, we aim to define a deformed symplectic form on $\mathscr{D}_{\mathcal{H}_A}$, instead of constructing the deformed theory at the level of the extended algebra of functions. In this way, the noncommutative background geometry is implemented directly at the level of the classical phase space structure on the horizon, which then extends to the associated Weyl algebra $\mathscr{W}_\Theta$ describing observables of the quantized theory on a (partially) noncommutative horizon. \\

To formulate a well-defined deformed symplectic structure for a fixed deformation parameter $\theta \in \mathbb{R}$, we consider the formal power series expansion \eqref{FiniteExpansionProduct} truncated at an arbitrary finite order $N \in \mathbb{N}$.

\begin{definition}
For $f,g\in\mathscr{D}_{\mathcal{H}_A}$ and $N\in\mathbb{N}$, we define the finite-order deformed symplectic form $\sigma_\Theta^{(N)}$ by
\begin{equation}
\label{DeformedSymplecticForm}
\sigma_\Theta^{(N)}(f,g) := \int_{\mathbb{R}_+\times\mathcal{S}} \left( (\partial_V f) \star_\Theta^{(N)} g - (\partial_V g) \star_\Theta^{(N)} f \right) \, dV \, d\mathrm{vol}_\mathcal{S},
\end{equation}

\noindent where $\star_\Theta^{(N)}$ denotes the $N$-th order expansion \eqref{FiniteExpansionProduct} of the deformed product \eqref{DilationProduct}.
\end{definition}

\noindent \textbf{Remark.} If we were to define the deformed symplectic form via the full nonperturbative deformed product \eqref{DilationProduct}, the integral over $\mathcal{H}_A$ may possibly diverge due to potentially insufficient $V$-decay of the integrand. However, restricting $\star_\Theta$ to an arbitrary finite order circumvents this issue and proves to be sufficient for the subsequent considerations.

\begin{lemma}\label{LemmaDeformedSymplecticInvariance}
The deformed symplectic form $\sigma^{(N)}_\Theta$ remains invariant under $\phi_t$, i.e.,
\begin{equation}
    \sigma^{(N)}_\Theta (\phi^\ast_t f , \phi^\ast_t g) = \sigma_\Theta^{(N)}(f,g)
\end{equation}

\noindent for all $f,g\in\mathscr{D}_{\mathcal{H}_A}$, where $\phi^\ast_t f(V,\vartheta,\varphi) = f(e^{-2\pi t}V,\vartheta,\varphi)$ denotes the pullback of $f\in\mathscr{D}_{\mathcal{H}_A}$ by $\phi_t$.
\end{lemma}

\noindent The proof for this statement is given in Appendix \ref{ProofDeformedSymplecticInvariance}. \\

For a fixed deformation parameter $\theta$, the bilinear form $\sigma_\Theta^{(N)}$ possibly exhibits a nontrivial null space. Consequently, in order to obtain a strictly nondegenerate symplectic form, we must work with a quotient space that factors out the kernel of $\sigma_\Theta^{(N)}$.

\begin{definition}
We define by
\begin{equation}
    \mathscr{E}_\Theta^{(N)} = \mathscr{D}_{\mathcal{H}_A} / \mathscr{K}_\Theta^{(N)} 
\end{equation}

\noindent the quotient space of $\mathscr{D}_{\mathcal{H}_A}$ with the kernel
\begin{equation}
    \mathscr{K}_\Theta^{(N)} = \{ f\in\mathscr{D}_{\mathcal{H}_A} \; \vert \; \sigma_\Theta^{(N)}(f,g) = 0 \quad\forall\;g\in\mathscr{D}_{\mathcal{H}_A} \}
\end{equation}

\noindent removed.
\end{definition}

\begin{lemma}
The bilinear form $\sigma_\Theta^{(N)}$ rigorously defines a symplectic form on $\mathscr{E}_\Theta^{(N)}$, i.e., it is antisymmetric, bilinear, and nondegenerate.
\end{lemma}

\textbf{Proof.} Antisymmetry of $\sigma_\Theta^{(N)}$ is a direct consequence of the integrand in \eqref{DeformedSymplecticForm} being antisymmetric under the exchange $f\leftrightarrow g$. Furthermore, bilinearity is inherited from the linearity of the integral and the partial derivative $\partial_V$, as well as the bilinearity of the map $(f,g) \mapsto f\star_\Theta g$, which in turn stems from the linearity of the group actions generating the deformed product. Finally, nondegeneracy holds by construction, since the quotient space $\mathscr{E}_\Theta^{(N)}$ does not contain any elements within the kernel $\mathscr{K}_\Theta^{(N)}$ of $\sigma_\Theta^{(N)}$, which ensures that $\sigma_\Theta^{(N)}$ is indeed strictly nondegenerate on the respective equivalence classes. $\qedsymbol$ \\

The resulting symplectic space $\left( \mathscr{E}_\Theta^{(N)} , \sigma_\Theta^{(N)} \right)$ allows us to consistently construct a deformed quantum theory on $\mathcal{H}_A$, up to any finite order $N$ in the deformation parameter $\theta$, by considering the associated Weyl algebra $\mathscr{W}_\Theta^{(N)}$ \cite{BR:1981opalg2}, whose elements $W_\Theta^{(N)}\in\mathscr{W}_\Theta^{(N)}$ fulfill the Weyl relations
\begin{align}
    W_\Theta^{(N)}(f)^\dagger &= W_\Theta^{(N)} (-f) \\
    W_\Theta^{(N)} (f) W_\Theta^{(N)} (g) &= e^{\frac{i}{2} \sigma_\Theta^{(N)}(f,g)} W_\Theta^{(N)} (f+g)
\end{align}

\noindent for all $f,g \in \mathscr{E}_\Theta^{(N)}$. We may define quasi-free states $\omega_\Theta^{(N)}$ on $\mathscr{W}_\Theta^{(N)}$ by
\begin{equation}
\omega_\Theta^{(N)} (W_\Theta(f)) := e^{-\frac{1}{2}\Lambda_\Theta^{(N)}(f,f)},
\end{equation}

\noindent where
\begin{equation}
    \Lambda_\Theta^{(N)} (f,g) := \mu_\Theta^{(N)} (f,g) + \frac{i}{2}\sigma_\Theta^{(N)}(f,g),
\end{equation}

\noindent for some bilinear form $\mu_\Theta^{(N)}$ on $\mathscr{E}_\Theta^{(N)}$ fulfilling
\begin{equation}
\label{DominatingInequality}
\mu_\Theta^{(N)} (f,f) \mu_\Theta^{(N)}(g,g) \geq \frac{1}{4}\left\vert \sigma^{(N)}_\Theta(f,g)\right\vert^2,
\end{equation}

\noindent for all $f,g\in\mathscr{E}_\Theta^{(N)}$, while the antisymmetric part $\sigma_\Theta(f,g)^{(N)}$ is exclusively defined by the underlying theory and thus the same for each quasi-free state. \\

Furthermore, in analogy to \cite{SV:1996tomitaki,KPV:2021ea}, the Weyl algebra $\mathscr{W}_\Theta^{(N)}$, generated by test functions $f \in \mathscr{E}_\Theta^{(N)}$ whose compact support is restricted to $\mathcal{H}_A^\mathcal{R}$, corresponds to the algebra of the right horizon portion $\mathcal{H}_A^\mathcal{R}$, so that we will henceforth denote it by $\mathscr{W}_\Theta^\mathcal{R}$. Note that all results also apply, \emph{mutatis mutandis}, to the "left" algebra $\mathscr{W}_\Theta^\mathcal{L}$. 

\begin{definition}\label{DefinitionDeformedState}
We define $\mathring{\mu}_\Theta^{(N)}$ to be the symmetric bilinear form on $\mathscr{E}_\Theta^{(N)}$ which satisfies \eqref{DominatingInequality}, and defines the quasi-free KMS state $\mathring{\omega}_\Theta^{(N)}$ on $\mathscr{W}_\Theta^\mathcal{R}$ with respect to the projected Killing flow $\phi_t: V \mapsto e^{2\pi t}V$ at inverse temperature $\beta_\mathrm{KMS} = \frac{2\pi}{\kappa}$, such that the corresponding GNS vector $\mathring{\Omega}_\Theta^{(N)}$ is not only cyclic, but also separating\footnote{We refer to \cite[Proposition 4.1]{Longo:2022mod} for the precise conditions for this property to hold.} for the von Neumann algebra $\mathring{\rho}^{(N)}_\Theta\left( \mathscr{W}_\Theta^\mathcal{R} \right)''$, where $\mathring{\rho}_\Theta^{(N)}$ denotes the GNS representation corresponding to the state $\mathring{\omega}_\Theta^{(N)}$.
\end{definition}

\noindent \textbf{Remark.} Following the arguments presented in \cite{HNS:1984scalim} and employing the KMS condition (cf. \cite{Kubo:1957kms,MS:1959kms,HHW:1967kms}), the two-point function $\mathring{\Lambda}^{(N)}_\Theta$ of the state $\mathring{\omega}_\Theta^{(N)}$ can be expressed in terms of the symplectic form $\sigma_\Theta^{(N)}$ as
\begin{equation}
\mathring{\Lambda}^{(N)}_\Theta (f,g) = \frac{i}{2\pi} \int_{\mathbb{R}\times\mathbb{R}} \frac{e^{\beta E}}{e^{\beta E}-1} \, \sigma_\Theta^{(N)}(\phi^\ast_t f,g) \, e^{iEt} \, dE \, dt.
\end{equation}

Given that both $\mathring{\mu}_\Theta^{(N)}$ and $\sigma_\Theta^{(N)}$ are invariant under $\phi_t$, it follows that $\phi_t$ is implemented in the GNS representation of $\mathring{\omega}_\Theta^{(N)}$ by the strongly continuous unitary action
\begin{equation}
\mathtt{U}^{\Theta,(N)}_{2\pi t}\,\mathring{\rho}^{(N)}_\Theta\big(W_\Theta^{(N)}(f)\big) = \mathring{\rho}^{(N)}_\Theta\big(W_\Theta^{(N)}(\phi_t^\ast f)\big).
\end{equation}

\noindent which leaves $\mathring{\Omega}_\Theta^{(N)}$ invariant. Since the modular group associated with a quasi-free KMS state is uniquely determined by the KMS condition and the modular data \cite{Summers:2006modth, BR:1981opalg2}, the modular flow for the deformed theory on $\mathcal{H}_A^\mathcal{R}$ must coincide with the projected Killing flow $\phi_t:V \mapsto e^{2\pi t}V$, which is equivalent to the statement
\begin{equation}
\label{DeformedModularOperator}
\Delta^{it}_{\mathcal{R},\Theta,(N)} = \;\mathtt{U}^{\Theta,(N)}_{2\pi t},
\end{equation}

\noindent for all $N\in\mathbb{N}$ and fixed deformation parameters $\theta$.

\section{Relative Entropy in the Deformed QFT}\label{SecDeformedRelativeEntropy}
In analogy to the semiclassical results derived in \cite{KPV:2021ea,Longo:2019relent,Hollands:2020relent,DAngelo:2021bhre,DM:2025sce}, our goal is to non-perturbatively compute the relative entropy between the state $\mathring{\omega}_\Theta^{(N)}$ and a coherent excitation
\begin{equation}
\omega^{(N),f}_\Theta (W_\Theta^{(N)}(g)) := \mathring{\omega}_\Theta^{(N)} \left( W_\Theta^{(N)}(-f) W_\Theta^{(N)}(g) W_\Theta^{(N)}(f) \right).
\end{equation}

\noindent To simplify the notation in the subsequent derivations, we henceforth omit the explicit mention of the order $N$ and instead denote the deformed product by $\star_\Theta$, the deformed symplectic form by $\sigma_\Theta$, and states on $\mathscr{W}_\Theta^\mathcal{R}$ by $\mathring{\omega}_\Theta$ or $\omega_\Theta^f$, respectively. However, we emphasize that all of these objects have to be implicitly understood with respect to the finite-order truncation discussed in the previous section. \\

In order to compute the relative entropy between the states $\mathring{\omega}_\Theta$ and $\omega_\Theta^f$, we recall the Araki-Uhlmann formula \cite{Araki:1976relent,Uhlmann:1977relent}
\begin{equation}
\label{ArakiUhlmann}
S^\mathrm{rel} \left( \mathring{\omega}_\Theta \, \Vert \, \omega^f_\Theta \right) = i \left. \frac{d}{dt} \right\vert_{t=0} \langle \Omega^f_\Theta
\vert\Delta^{it}_\Theta \Omega^f_\Theta\rangle,
\end{equation}

\noindent where $\Omega^f_\Theta = \mathring{\rho}_\Theta\left(W(f)\right)\, \mathring{\Omega}_\Theta$, and $\Delta_\Theta$ denotes the modular operator for $\mathring{\omega}_\Theta$, as given in Equation \eqref{DeformedModularOperator}. In light of this, the relative entropy between coherent states for $\mathscr{W}_\Theta^\mathcal{R}$ reduces to \cite{KPV:2021ea, HI:2019new}
\begin{equation}
\label{ArakiUhlmannCoherentSimplified}
S^\mathrm{rel} \left( \mathring{\omega}_\Theta \, \Vert \, \omega^f_\Theta \right) = i \left. \frac{d}{dt} \right\vert_{t=0} \mathring{\omega}_\Theta \left( W_\Theta(-f) W_\Theta(f^t) \right),
\end{equation}

\noindent where $f^t$ is shorthand notation for
\begin{equation}
f^t (V,\vartheta,\varphi) := \phi_t^\ast f(V , \vartheta,\varphi) = f(e^{-2\pi t} V , \vartheta,\varphi).
\end{equation}

\begin{theorem}
The relative entropy between coherent states in the deformed theory is given by
\begin{equation}
S^\mathrm{rel} \left( \mathring{\omega}_\Theta \, \Vert \, \omega^f_\Theta \right) = \pi \int_{\mathbb{R}_+\times\mathcal{S}}  \left( (\partial_V f) \star_\Theta (V(\partial_V f)) - \left((\partial_V f) +(V(\partial^2_V f)) \right) \star_\Theta f\right) \, dV d\mathrm{vol}_\mathcal{S}.
\label{DeformedRelativeEntropy}
\end{equation}
\end{theorem}

\noindent \textbf{Proof.} Following the derivation in \cite{KPV:2021ea}, we find that the relative entropy \eqref{ArakiUhlmannCoherentSimplified} explicitly writes as
\begin{align}
S^\mathrm{rel} \left( \mathring{\omega}_\Theta \, \Vert \, \omega^f_\Theta \right) &= i \left. \frac{d}{dt} \right\vert_{t=0} e^{\frac{i}{2}\sigma_\Theta(f,f^t)} e^{-\frac{1}{2}\Lambda_\Theta(f^t - f, f^t - f)} \\
&= \frac{1}{2} \left. \frac{d}{dt} \right\vert_{t=0} \sigma_\Theta (f^t,f),
\label{SigmaDerivative}
\end{align}

\noindent where we have used that $\sigma_\Theta(f,f) = 0$ due to its antisymmetry, together with the relation $\left. \frac{d}{dt}\right\vert_{t=0} \Lambda_\Theta(f^t-f, f^t-f) = 0$, which follows from the linearity of $\Lambda_\Theta$. More precisely, Formula \eqref{SigmaDerivative} is computed by
\begin{align}
S^\mathrm{rel} \left( \mathring{\omega}_\Theta \, \Vert \, \omega^f_\Theta \right) &= \frac{1}{2} \left. \frac{d}{dt} \right\vert_{t=0} \int_{\mathbb{R}_+\times\mathcal{S}} e^{i\theta\left(X\partial_X \partial_\eta - Y\partial_Y \partial_\xi \right)} \big( \partial_X f^t(X,\vartheta,\xi) \, f(Y,\vartheta,\eta) \\
&\hspace{5.8cm}- \partial_X f(X,\vartheta,\xi) \, f^t(Y,\vartheta,\eta)\big) \bigg\vert_{\substack{X=Y=V \\ \xi=\eta=\varphi}} \,dV d\mathrm{vol}_\mathcal{S}. \nonumber
\end{align}

\noindent Recalling that the integrand is compactly supported up to any finite order in $\theta$, we may act the $t$-derivative on the integrand, in order to obtain
\begin{align}
\label{RelEntStep1}
S^\mathrm{rel} \left( \mathring{\omega}_\Theta \, \Vert \, \omega^f_\Theta \right) &= \frac{1}{2}  \int_{\mathbb{R}_+\times\mathcal{S}} e^{i\theta\left(X\partial_X \partial_\eta - Y\partial_Y \partial_\xi \right)} \big( \partial_X \dot{f}^0(X,\vartheta,\xi) \, f(Y,\vartheta,\eta) \\
&\hspace{4.8cm}- \partial_X f(X,\vartheta,\xi) \, \dot{f}^0(Y,\vartheta,\eta)\big) \bigg\vert_{\substack{X=Y=V \\ \xi=\eta=\varphi}} \,dV d\mathrm{vol}_\mathcal{S}, \nonumber
\end{align}

\noindent where $\dot{f}^0$ is shorthand notation for $\left. \tfrac{d}{dt} \right\vert_{t=0} f^t$. Explicitly, we have
\begin{align}
\dot{f}^0(V) &= -2\pi V (\partial_V f)(V), \\
\partial_V \dot{f}^0(V) &= -2\pi \left( (\partial_V f)(V) + V(\partial_V^2f)(V) \right),
\end{align}

\noindent so that the relative entropy \eqref{RelEntStep1} computes as
\begin{align}
\label{RelEntStep2}
S^\mathrm{rel} \left( \mathring{\omega}_\Theta \, \Vert \, \omega^f_\Theta \right) &= \pi \int_{\mathbb{R}_+\times\mathcal{S}}  e^{i\theta\left(X\partial_X \partial_\eta - Y\partial_Y \partial_\xi \right)} \Big( \partial_X f(X,\vartheta,\xi) \, Y(\partial_Y f)(Y,\vartheta,\eta) \\
&\hspace{1.37cm} - \left( (\partial_X f)(X,\vartheta,\xi) + X(\partial_X^2 f)(X,\vartheta,\xi) \right) f(Y,\vartheta,\eta) \Big) \bigg\vert_{\substack{X=Y=V \\ \xi=\eta=\varphi}} \,dV d\mathrm{vol}_\mathcal{S}, \nonumber
\end{align}

\noindent where we can identify the individual formal star products $\star_\Theta$, in order to finally obtain Formula \eqref{DeformedRelativeEntropy}. $\qedsymbol$ \\

\noindent \textbf{Remark.} In contrast to the usual pointwise product between test functions $f\in \mathscr{D}_{\mathcal{H}_A}$, there exists no general formula for integration by parts with respect to the deformed product $\star_\Theta$, so that the general Formula \eqref{DeformedRelativeEntropy} for the relative entropy in the deformed theory, does not further simplify at arbitrary orders in $\theta$. In fact, if certain conditions are fulfilled, e.g., if the deformed product is purely defined via spatiotemporal translations with a constant deformation parameter $\theta$, such as the Rieffel product \eqref{RieffelProduct}, integration by parts is actually possible with respect to deformed products. However, these conditions are not fulfilled by the deformed product $\star_\Theta$ studied in this work. \\

Subsequently, we restrict ourselves to the second-order expansion of the relative entropy in the deformed theory, for which we employ Formula \eqref{SecondOrderProduct}, in order to obtain a much simpler expression, as the following corollary entails.

\begin{corollary}\label{CorollarySecondOrderEntropy}
The relative entropy between coherent states in the deformed theory is, up to second-order in the deformation parameter $\theta$, given by
\begin{align}
\label{SecondOrderDeformedRelativeEntropy}
S^\mathrm{rel} \left( \mathring{\omega}_\Theta \, \Vert \, \omega^f_\Theta \right) &= 2\pi\int_{\mathbb{R}_+\times\mathcal{S}} V \left( (\partial_V f)^2 + \frac{\theta^2}{2} \left( \partial_V \partial_\varphi f \right)^2 \right) \, dV d\mathrm{vol}_\mathcal{S} +\mathcal{O}(\theta^3), \\
&= S^\mathrm{rel}(\mathring{\omega}\,\Vert\,\omega^f) + \pi\theta^2 \int_{\mathbb{R}_+\times\mathcal{S}} V \left( \partial_V \partial_\varphi f \right)^2 \, dV d\mathrm{vol}_\mathcal{S} +\mathcal{O}(\theta^3),
\end{align}

\noindent where $S^\mathrm{rel}(\mathring{\omega}\,\Vert\,\omega^f)$ denotes the entropy between coherent states in the (semi-)classical horizon theory (cf. \cite{KPV:2021ea}).
\end{corollary}

\noindent The proof for this statement is given in Appendix \ref{ProofSecondOrderEntropy}. \\

\noindent Given that $V$ is strictly positive on $\mathcal{H}_A^\mathcal{R}\supset\mathrm{supp}(f)$, our result shows that the relative entropy between coherent states on $\mathscr{W}^\mathcal{R}_\Theta$ is strictly positive, at least up to second order in the deformation parameter $\theta$. This is in agreement with standard results from algebraic QFT and quantum information theory \cite{Araki:1976relent, Uhlmann:1977relent, Wehrl:1978ent, OP:1993ent, HI:2019new}, where positivity is established as a fundamental property of the relative entropy in operator algebraic frameworks. \\

In particular, we observe that the (second-order) relative entropy \eqref{SecondOrderDeformedRelativeEntropy} in the deformed theory admits a natural interpretation in terms of the localization properties of the coherent excitation $\omega^f_\Theta$ on $\mathcal{H}_A$. The zeroth-order term, corresponding to the undeformed relative entropy, increases with increasingly sharp localization of the test function $f\in\mathscr{D}_{\mathcal{H}_A}$ along the null generators parametrized by the affine coordinate $V$. By contrast, the deformation-induced term at order $\theta^2$, describing newly arising angular correlations, not only involves derivatives in $V$, but also in the angular coordinate $\varphi$, and therefore increases with increasingly sharp localization in both the null and angular directions. Consequently, the relative entropy quantifies the distinguishability between the reference state $\mathring{\omega}_\Theta$ and its coherent excitation $\omega^f_\Theta$ in terms of the localization properties of the test function $f$ on $\mathcal{H}_A$.

\subsection{A Relative Entropy Area Law in the Deformed Theory}
Following the arguments presented in \cite{KPV:2021ea} (see also \cite{DAngelo:2021bhre}), we aim to investigate the proportionality between the relative entropy for coherent excitations in the deformed theory and the surface area of the horizon cross-section. Therefore, let us consider $f,g \in \mathscr{D}_{\mathcal{H}_A}$, such that $f$ is compactly supported in $I\times\Sigma_f \subset \mathbb{R}_+\times\mathcal{S}$ and $g$ is supported in $I\times\Sigma_g \subset \mathbb{R}_+\times\mathcal{S}$, where $\Sigma_f,\Sigma_g \subset\mathcal{S}$ are chosen to be disjoint subsets of $\mathcal{S}$, i.e., $\Sigma_f \cap \Sigma_g = \emptyset$. Defining the function $h:= (f+g)\in \mathscr{D}_{\mathcal{H}_A}$, let us consider the coherent state $\omega_\Theta^h$, and compute the corresponding relative entropy by
\begin{equation}
S^\mathrm{rel} \left( \mathring{\omega}_\Theta \, \Vert \, \omega^h_\Theta \right) = \pi \int_{\mathbb{R}_+\times\mathcal{S}} \left( (\partial_V h) \star_\Theta (V(\partial_V h)) - \left((\partial_V h) +(V(\partial^2_V h)) \right) \star_\Theta h\right) \, dV d\mathrm{vol}_\mathcal{S}.
\end{equation}

\noindent More precisely, the relative entropy takes the explicit form
\begin{align}
S^\mathrm{rel} \left( \mathring{\omega}_\Theta \, \Vert \, \omega^h_\Theta \right) &= \pi \int_{\mathbb{R}_+\times\mathcal{S}} \Bigl( \left( \partial_V f + \partial_V g \right)\star_\Theta \left(V \partial_V f + V\partial_V g \right) \\
&\hspace{2.2cm} - \left( \left( \partial_Vf + \partial_V g \right) + V \partial_V^2 f + V\partial_V^2 g \right) \star_\Theta (f+g) \Bigl) \, dV d\mathrm{vol}_\mathcal{S}, \nonumber
\end{align}

\noindent which computes to
\begin{align}
S^\mathrm{rel} \left( \mathring{\omega}_\Theta \, \Vert \, \omega^h_\Theta \right) &= \pi \int_{\mathbb{R}_+\times\mathcal{S}} \Bigl( (\partial_V f) \star_\Theta (V\partial_V f) + (\partial_V g) \star_\Theta (V\partial_V g) \label{AddRelEntExp} \\
&\hspace{2.2cm} + (\partial_V g) \star_\Theta (V\partial_V f) + (\partial_V f) \star_\Theta (V\partial_V g) \nonumber  \\
&\hspace{2.2cm} - (\partial_V f + V\partial_V^2 f) \star_\Theta f - (\partial_Vg + V\partial_V^2 g) \star_\Theta f\nonumber \\
&\hspace{2.2cm} - (\partial_V f + V\partial_V^2 f) \star_\Theta g - (\partial_Vg + V\partial_V^2 g) \star_\Theta g \Bigl) \, dV d\mathrm{vol}_\mathcal{S}, \nonumber
\end{align}

\noindent where we have used the distributivity of sums of functions in $\mathscr{D}_{\mathcal{H}_A}$ over the deformed product $\star_\Theta$, which is straightforward to prove using the definition of $\star_\Theta$. Identifying the individual terms in the integrand of expression \eqref{AddRelEntExp}, we find that the relative entropy between the state $\mathring{\omega}_\Theta$ and its coherent excitation $\omega_\Theta^h$ is given by the sum of the individual relative entropies $S^\mathrm{rel} ( \mathring{\omega}_\Theta \, \Vert \, \omega^f_\Theta )$ and $S^\mathrm{rel} ( \mathring{\omega}_\Theta \, \Vert \, \omega^g_\Theta )$, plus an additional contribution, i.e., we obtain
\begin{align}
\label{AddRelEntFinal}
S^\mathrm{rel} \left( \mathring{\omega}_\Theta \, \Vert \, \omega^h_\Theta \right) &= S^\mathrm{rel} \bigl( \mathring{\omega}_\Theta \, \Vert \, \omega^f_\Theta \bigl) + S^\mathrm{rel} \bigl( \mathring{\omega}_\Theta \, \Vert \, \omega^g_\Theta \bigl) \\
&\hspace{0.5cm} + \pi\int_{\mathbb{R}_+\times\mathcal{S}} \Bigl( (\partial_V g) \star_\Theta (V\partial_V f)  +  (\partial_V f) \star_\Theta (V\partial_V g) \nonumber \\
&\hspace{2.7cm} - (\partial_V g + V\partial_V^2g) \star_\Theta f - (\partial_V f + V\partial_V^2 f) \star_\Theta g \Bigl) dV d\mathrm{vol}_\mathcal{S}. \nonumber
\end{align}

However, according to Lemma \ref{LemmaSupportPreservation}, the deformed product $\star_\Theta$ preserves angular localization up to any finite order in $\theta$. In particular, if the angular supports of $f,g\in\mathscr{D}_{\mathcal{H}_A}$ are disjoint, then all mixed contributions in expression \eqref{AddRelEntFinal} vanish identically up to any finite order in $\theta$, since in this case, the regions $\Sigma_f$ and $\Sigma_g$ do not have any overlap by assumption. Consequently, the additivity of the relative entropy with respect to angular separation, as derived in \cite{KPV:2021ea}, holds to arbitrary but finite order in $\theta$. In contrast, for the full nonperturbative product defined by \eqref{DilationProduct}, the product of two functions with disjoint angular supports need not remain angularly localized, as depicted in Figure \ref{FigureSupportExtension}, so that the mixed contributions generally do not vanish, leading to a substantially more involved behaviour of the relative entropy. \\

On the other hand, if one considers coherent excitations with respect to test functions that are constant over the horizon cross-section $\mathcal{S}$\footnote{At this point, we like to recall that, strictly speaking, these functions are not elements of the function space $\mathscr{D}_{\mathcal{H}_A}$, since they cannot be non-zero, compactly supported within $\mathbb{R}_+\times\mathcal{I}_\vartheta\times(-\pi,\pi)\subset\mathbb{R}_+\times\mathbb{R}^2$, \emph{and} smooth at the boundaries of the subset $\mathcal{I}_\vartheta\times(-\pi,\pi)$. However, in the specific case of the relative entropy in the deformed theory, where we consider the deformed product of such a function with derivatives of itself, the deformed product \eqref{DilationProduct} (and hence also \eqref{FormalExpansionProduct}) remains well-defined and physically meaningful.}, say, some $\mathfrak{f}\in C^\infty_0(\mathbb{R}_+)$, the proportionality between the relative entropy and the surface area of $\mathcal{S}$ holds in the following sense. Assuming this particular case, we observe that the integral \eqref{AddRelEntExp} becomes independent of the coordinates $(\vartheta,\varphi)$, so that the relative entropy computes, almost analogously to the situation in Theorem \ref{TheoremUnitality}, as
\begin{align}
S^\mathrm{rel} \left( \mathring{\omega}_\Theta \, \Vert \, \omega^\mathfrak{f}_\theta \right) &= \pi \int_{\mathbb{R}_+\times\mathcal{S}} \left( (\partial_V \mathfrak{f})\star_\Theta (V(\partial_V \mathfrak{f})) - \left((\partial_V \mathfrak{f}) +(V(\partial^2_V \mathfrak{f})) \right) \star_\Theta \mathfrak{f} \right) \, dV d\mathrm{vol}_\mathcal{S}\nonumber\\
&= \pi \int_\mathcal{S} d\mathrm{vol}_\mathcal{S} \int_{\mathbb{R}_+} \left( V (\partial_V \mathfrak{f})^2 - (\partial_V \mathfrak{f})\, \mathfrak{f} - V(\partial_V^2 \mathfrak{f}) \, \mathfrak{f} \right) \, dV \nonumber \\
&= 2\pi \mathcal{A}(\mathcal{S}) \int_{\mathbb{R}_+} V (\partial_V \mathfrak{f})^2 dV,
\end{align}

\noindent where we have used that only the terms of order zero survive, and integrated by parts with respect to $V$, in the same way as in the proof of Corollary \ref{CorollarySecondOrderEntropy}. Furthermore, we have used that the integrand is independent of the variables $(\vartheta,\varphi)$, so that integration over the volume element $d\mathrm{vol}_\mathcal{S}$ gives rise to the surface area $\mathcal{A}(\mathcal{S})$ of the horizon cross-section $\mathcal{S}$. \\

In fact, in this specific situation, we recover the same relative entropy as in the undeformed theory \cite{KPV:2021ea}. This follows directly from the unitality of the deformed product established in Theorem \ref{TheoremUnitality}, which carries over trivially to the expansion Formula \eqref{FiniteExpansionProduct}. As a result, (partially) constant functions are left invariant by the deformation, and the undeformed theory is recovered in this sector. Consequently, one again obtains a direct proportionality between the relative entropy between coherent states in the horizon theory and the corresponding geometric surface area $A(\mathcal{S}_r)$ of the horizon cross-section $\mathcal{S}_r$ at radius $r$, as derived in \cite{KPV:2021ea}.

\section{Conclusions}\label{Conclusions}
In this work, we have developed an algebraic deformation procedure for QFT on bifurcate Killing horizons in stationary and axisymmetric spacetimes. The construction is based on a pair of commuting spacetime symmetries, namely the flow generated by the timelike Killing vector field $\xi^a$, whose restriction to the horizon acts as affine dilations along the null generators, and the rotational symmetry around the axis generated by the axial Killing vector field $\psi^a$. The fact that the flows generated by these symmetries commute, and hence define a jointly Abelian group action on the horizon geometry, ensures the formal associativity of the resulting star product, i.e., associativity up to any finite order in the deformation parameter $\theta$. More precisely, the deformation acts on smooth compactly supported functions on the horizon and yields a noncommutative product that intertwines the affine null coordinate $V$ along the horizon with the azimuthal angular coordinate $\varphi$. Based on this product, we constructed a deformed symplectic space together with its associated Weyl algebra of observables. The resulting deformed theory provides a consistent formulation of QFT on a partially noncommutative horizon that, up to arbitrary finite order in $\theta$, preserves some essential structural properties of the undeformed setting. In particular, the deformed symplectic form remains invariant under the projected Killing flow, so that the modular data of the algebra of observables localized on the right portion of the horizon coincides with that of the undeformed theory. More generally, when considering the deformation up to any finite order in $\theta$, the localization properties of horizon observables are preserved. \\

The main results of this paper consist of a general formula for the relative entropy between coherent states in the deformed theory, together with an explicit expansion of this relative entropy up to second order in the deformation parameter $\theta$. In particular, the latter expression is manifestly nonnegative and exhibits explicit correction terms induced by the noncommutative background geometry which is implemented via the deformation procedure. More precisely, we found that the first-order contribution vanishes identically as a consequence of the geometric structure of the Poisson bracket underlying the deformation, while a strictly positive term arises at second order, given in terms of mixed derivatives in the null and angular directions. This yields a concrete perturbative characterization of how a partially noncommutative horizon geometry modifies the information content of coherent excitations with respect to their localization properties. \\

Assuming that the relative entropy provides the appropriate quantum field theoretic notion of entropy for matter fields on a fixed black hole background geometry, these results allow for a qualitative connection to the Page curve \cite{Page:1993og,Page:2013rev}. Since, in the purely semiclassical picture, the relative entropy is proportional to the horizon area \cite{KPV:2021ea}, the gradual decrease of the horizon area under black hole evaporation corresponds to a decreasing relative entropy, thereby giving rise to a Page curve for this quantity.\footnote{We also refer to \cite{AMMZ:2020} for a recent discussion of the Page curve from a holographic perspective.} In particular, considering regimes where the horizon area becomes comparatively small, so that the deformation parameter, which is assumed to be of the order of the Planck length, is no longer negligible, the second-order corrections derived in this work lead to an upward correction of the resulting Page curve at late times. \\

Beyond these perturbative results, it is instructive to revisit the structural picture of horizon algebras outlined towards the end of Section \ref{QFTonKillingHorizons}. In the undeformed setting, the horizon theory on $\mathcal{H}_A^\mathcal{R}$ may be interpreted as admitting a transverse decomposition into identical chiral CFTs propagating along the null generators, in analogy to \cite{MTW:2022npa}. In contrast, the deformation constructed here is generated by affine dilations in $V$ together with shifts in the azimuthal coordinate $\varphi$, and thus acts nontrivially in the transverse direction. Accordingly, although the modular flow of the deformed theory remains geometric, the deformation is generally expected not to preserve the decomposable structure of the horizon algebra, but instead to induce correlations between observables localized on distinct null rays. \\

Furthermore, our results point to a nontrivial feature of the deformation. Coherent excitations that are uniformly distributed along the azimuthal direction, corresponding to angular zero modes, remain unaffected by the noncommutative structure implemented by the star product. On the other hand, coherent excitations with nontrivial angular dependence, i.e., excitations that are localized in the angular sector, are sensitive to the effects of the deformation. In this sense, the deformation couples specifically to angularly localized degrees of freedom, while leaving delocalized coherent excitations that are constant in the azimuthal coordinate invariant. \\

Notably, the same qualitative behaviour is observed in a companion work by the authors, to appear elsewhere \cite{DMV:2026def}, where a generalized Rieffel deformation \cite{Much:2025genrd} of the $2$-sphere is studied, corresponding to an effective implementation of the fuzzy sphere \cite{Madore:1992fs} in the background of a near-horizon QFT. In that setting, the lowest angular mode, i.e., the zero-mode, which is uniformly distributed over the sphere, is likewise insensitive to the deformation procedure, whereas higher angular modes with increasing angular localization do respond to the noncommutative background geometry more clearly. We argue that this parallel behaviour is not incidental, but rather reflects the general property of unitality of star product deformations. In this sense, the insensitivity of angular zero modes can be viewed as a characteristic feature of deformation based approaches to a noncommutative horizon geometry. \\

A complementary physical intuition for this behaviour stems from the fact that restricting a QFT to a Killing horizon does not yield a theory that is subject to a fully dynamical evolution. Instead, the horizon rather serves as a lightlike initial data surface that determines the quantum field in its domain of dependence. This interpretation is closely aligned with the analysis of Kay and Wald \cite{KW:1991hadamard}, who showed that for stationary Hadamard states on spacetimes with a bifurcate Killing horizon, a given state is already uniquely determined on a large subalgebra generated by solutions that fall entirely through the horizon. In this sense, all the necessary information about the state is encoded in the initial data on horizon data itself. From this perspective, the restriction of a QFT to a Killing horizon naturally compares to a quantum system on an equal time hypersurface, where correlations are determined by a given state together with the canonical commutation relations, rather than emerging from genuine time evolution. \\

More generally, we would like to remark that, strictly speaking, the assumption of stationarity is not necessarily essential for this work. It only enters in the derivation of the specific form of the two-point function \eqref{UniversalScalingLimit} on $\mathcal{H}_A^{\mathcal{R}}$, which is thermal at Hawking temperature with respect to dilations. However, analogous two-point functions which are thermal with respect to dilations exist in arbitrary spherically symmetric spacetimes, even if they are dynamical \cite{KPV:2021ea}. Hence, if one were to obtain a similar two-point function on the horizon which possibly depends on the coordinate $\vartheta$, as long as it is thermal with respect to dilations in a more general dynamical axisymmetric spacetime, one could straightforwardly apply the results of this paper to such a situation. A promising class of spacetimes where this may be achieved in the near future is given by the Kerr-Vaidya(-de Sitter) family of spacetimes (cf. \cite{DT:2020kv}), which were recently found to also possess Kodama-like vector fields \cite{DV:2024kod}, which are a crucial ingredient in the construction of thermal states near geometric horizons. \\

Finally, we emphasize that the present construction should be viewed as a mathematical toy model rather than a complete effective description of quantum gravitational phenomena. The deformation intertwines only those spacetime directions that correspond to global spacetime symmetries, namely the affine null coordinate along the horizon generators and the azimuthal angle associated with the axial Killing vector field, while leaving the remaining spacetime dimensions, i.e., the one parametrized by the coordinate $\vartheta$, entirely classical. A fully developed noncommutative spacetime model relevant for quantum gravity would require a deformation that consistently incorporates all spacetime degrees of freedom. Although such a construction lies beyond the scope of the present work, the framework developed here provides a concrete and technically feasible setting in which the qualitative effects of symmetry based star product deformations can be studied explicitly.

\appendix
\section{Proofs}
In this appendix, we collect the proofs that are too extensive for the main text.

\subsection{Proof for Proposition \ref{PropositionExpansion}}\label{ProofPropositionExpansion}
Let us write out the deformed product \eqref{DilationProduct}, as before, i.e.,
\begin{equation}
(f\star_\Theta g)(V,\vartheta,\varphi) = \frac{1}{(2\pi)^2} \int_{\mathbb{R}^2\times\mathbb{R}^2} e^{-i(ab+\alpha\beta)} f(e^{-\theta\alpha}V,\vartheta,\varphi+\theta a) \, g(e^bV,\vartheta,\varphi+\beta) \, da \, db \, d\alpha \, d\beta
\end{equation}

\noindent Next, we use that dilations $\mathfrak{D}$ are generated by $V\partial_V$ and that rotations $\mathfrak{L}$ about the axis of symmetry take the form of shifts in $\varphi$, and are thus generated by $\partial_\varphi$, i.e.,
\begin{align}
f(e^c V,\vartheta,\varphi) &= e^{c (V\partial_V)} f(V,\vartheta,\varphi) \\
f(V,\vartheta,\varphi+\gamma) &= e^{\gamma \partial_\varphi} f(V,\vartheta,\varphi),
\end{align}

\noindent so that we may expand the integral in terms of the formal power series
\begin{align}
\label{ExpansionStep1}
(f\star_\Theta g)(x) &= \frac{1}{(2\pi)^2} \int_{\mathbb{R}^2} da\, db \, e^{-iab} \int_{\mathbb{R}^2} d\alpha\,d\beta \, e^{-i\alpha\beta} \\
&\hspace{1.5cm} \cdot \sum_{j,l,m,n\in\mathbb{N}_0} \left( \frac{(-\theta\alpha)^j}{j!}(V\partial_V)^j \frac{(\theta a)^l}{l!}(\partial_\varphi)^l f(x) \, \frac{b^m}{m!} (V\partial_V)^m \frac{\beta^n}{n!}(\partial_\varphi)^n g(x) \right), \nonumber
\end{align}

\noindent where $x=(V,\vartheta,\varphi)\in \mathbb{R}_+\times\mathcal{S}$. Given that we are dealing with formal manipulations up to arbitrary but finite orders in $\theta$, we may rewrite Expression \eqref{ExpansionStep1} as
\begin{align}
\label{ExpansionStep2}
(f\star_\Theta g) &= \frac{1}{(2\pi)^2}\sum_{j,l,m,n\in\mathbb{N}_0} \frac{(-\theta)^j \theta^l}{j!\,l!\,m!\,n!} \left( (V\partial_V)^j \partial_\varphi^l f \right)\left( (V\partial_V)^m \partial_\varphi^n g \right) \\
&\hspace{3cm}\cdot\int_{\mathbb{R}^2} da\, db \, a^l b^m e^{-iab} \int_{\mathbb{R}^2} d\alpha\,d\beta \, \alpha^j \beta^n e^{-i\alpha\beta}. \nonumber
\end{align}

\noindent Moreover, we use the distributional identity
\begin{equation}
\int_\mathbb{R} A^n e^{-iAB} dA = 2\pi i^n \partial_B^n\,\delta(B),
\end{equation}

\noindent which is straightforwardly obtained by differentiating the Fourier representation of the Dirac $\delta$-distribution $n$ times ($n\in\mathbb{N}$). Similarly\footnote{We remark that this identity has to be understood in the distributional sense, i.e., one (implicitly) integrates against a test function, which in this case of the deformed product \eqref{DilationProduct} is given by (the norm limit of) the Schwartz function $\chi\in\mathscr{S}(\mathbb{R}^2 \times \mathbb{R}^2)$.}, we obtain
\begin{align}
\int_{\mathbb{R}^2} A^n B^m e^{-iAB} dA\, dB &= 2\pi i^n \int_\mathbb{R} B^m \partial^n_B \,\delta(B) \, dB \\
&= 2\pi i^n (-1)^n \int_\mathbb{R} (\partial_B^n B^m)\,\delta(B)\, dB  \\
&=2\pi (-i)^n n! \,\delta_{nm},
\end{align}

\noindent where we have performed an $n$-fold integration by parts in the second step. Applied to expression \eqref{ExpansionStep2}, this leads us to
\begin{align}
\label{ExpansionStep3}
(f\star_\Theta g) &= \sum_{j,l,m,n\in\mathbb{N}_0} \frac{(-\theta)^j \theta^l}{j!\,l!\,m!\,n!} \left( (V\partial_V)^j \partial_\varphi^l f \right)\left( (V\partial_V)^m \partial_\varphi^n g \right) (-i)^{m+n} m!\,n!\, \delta_{lm}\,\delta_{jn},
\end{align}

\noindent which consequently simplifies to
\begin{align}
\label{ExpansionStep4}
(f\star_\Theta g) &= \sum_{m,n\in\mathbb{N}_0} \frac{(i\theta)^n (-i\theta)^m}{m!\,n!} \left( (V\partial_V)^n \partial_\varphi^m f \right)\left( (V\partial_V)^m \partial_\varphi^n g \right).
\end{align}

\noindent Treating $f$ and $g$ as functions of independent variables during differentiation and identifying the variables afterwards, and introducing the index $k:=m+n$, we can further reformulate this expression as
\begin{align}
\label{ExpansionStep5}
(f\star_\Theta g) &= \sum_{k\in\mathbb{N}_0} \sum_{n=0}^k \frac{(i\theta)^n (-i\theta)^{k-n}}{n!\,(k-n)!} (X\partial_X)^n (Y\partial_Y)^{k-n} \partial_\xi^{k-n} \partial_\eta^n f(X,\xi)\, g(Y,\eta) \bigg\vert_{\substack{X=Y=V \\ \xi=\eta=\varphi}} \\
%&= \sum_{k\in\mathbb{N}_0} (i\theta)^k \sum_{n=0}^k \frac{1}{n!\,(k-n)!} (X\partial_X \partial_\eta)^n (-Y\partial_Y \partial_\xi)^{k-n} f(X,\xi)\, g(Y,\eta) \bigg\vert_{\substack{X=Y=V \\ \xi=\eta=\varphi}} \\
&= \sum_{k\in\mathbb{N}_0} \frac{(i\theta)^k}{k!} \sum_{n=0}^k {k\choose n} (X\partial_X \partial_\eta)^n (-Y\partial_Y \partial_\xi)^{k-n} f(X,\xi)\, g(Y,\eta) \bigg\vert_{\substack{X=Y=V \\ \xi=\eta=\varphi}},
\label{ExpansionStep6}
\end{align}

\noindent which by the binomial theorem simplifies to the formal power series

\begin{equation}
\label{ExpansionStep7}
(f\star_\Theta g) = \sum_{k\in\mathbb{N}_0} \frac{(i\theta)^k}{k!} \left(X\partial_X \partial_\eta -Y\partial_Y \partial_\xi\right)^k f(X,\xi)\, g(Y,\eta) \bigg\vert_{\substack{X=Y=V \\ \xi=\eta=\varphi}}.
\end{equation}

\noindent Ultimately, this series can formally be written as
\begin{align}
(f\star_\Theta g)(V,\vartheta,\varphi) = \left. e^{i\theta\left( X\partial_X \partial_{\eta} - Y\partial_{Y} \partial_\xi \right)} f(X,\vartheta,\xi) \, g(Y,\vartheta,\eta) \right\vert_{\substack{X=Y=V \\ \xi=\eta=\varphi}},
\end{align}

\noindent i.e., as a bi-pseudodifferential operator acting on test functions $f,g\in\mathscr{D}_{\mathcal{H}_A}$ evaluated at a single point $(V,\vartheta,\varphi)\in\mathbb{R}_+\times\mathcal{S}$. $\qedsymbol$

\subsection{Proof for Theorem \ref{TheoremAssociativity}}\label{ProofAssociativity}
Throughout this proof, we suppress the dependence of the functions $f,g,h\in\mathscr{D}_{\mathcal{H}_A}$ on the coordinate $\vartheta$, which is unaffected by the deformation. Moreover, we introduce three distinct sets of horizon coordinates $(X,\xi)$, $(Y,\eta)$, and $(Z,\zeta)$ on $\mathcal{H}_A^\mathcal{R} \cong \mathbb{R}_+ \times \mathcal{S}$, and define the corresponding differential operators
\begin{align}
D_{XY} &= X\partial_X \partial_\eta - Y\partial_Y \partial_\xi, \\
D_{YZ} &= Y\partial_Y \partial_\zeta - Z\partial_Z \partial_\eta, \\
D_{XZ} &= X\partial_X \partial_\zeta - Z\partial_Z \partial_\xi.
\end{align}

\noindent Recalling Formula \eqref{FormalExpansionProduct}, we define
\begin{equation}
(f\star_\Theta g) (W,\varpi) = \left. e^{i\theta D_{XY}} f(X,\xi) g(Y,\eta) \right\vert_{\substack{X=Y=W \\ \xi=\eta=\varpi}} =: \left. F(X,\xi,Y,\eta) \right\vert_{\substack{X=Y=W \\ \xi=\eta=\varpi}} =: F (W,\varpi),
\end{equation}

\noindent where, crucially, $e^{i\theta D_{XY}}$ has to be understood as the formal power series
\begin{equation}
e^{i\theta D_{XY}} = \sum_{k\in\mathbb{N}_0}\frac{(i\theta)^k}{k!} D_{XY}^k = \sum_{k\in\mathbb{N}_0}\frac{(i\theta)^k}{k!}\left(X\partial_X \partial_\eta - Y\partial_Y \partial_\xi \right)^k.
\end{equation}

\noindent Next, we compute
\begin{align}
\label{TripleExpansionProduct_fg_h}
\big( (f\star_\Theta g) \star_\Theta h \big)(V,\varphi) &= \left. e^{i\theta\left(W\partial_W \partial_\zeta - Z\partial_Z \partial_\varpi \right)} F (W,\varpi) h(Z,\zeta) \right\vert_{\substack{W=Z=V \\ \varpi=\zeta=\varphi}} \\
&= \left. e^{i\theta\left(W\partial_W \partial_\zeta - Z\partial_Z \partial_\varpi \right)} \left( \left. e^{i\theta D_{XY}} f(X,\xi) g(Y,\eta) \right\vert_{\substack{X=Y=W \\ \xi=\eta=\varpi}} \right) h(Z,\zeta) \right\vert_{\substack{W=Z=V \\ \varpi=\zeta=\varphi}}.
\nonumber
\end{align}

\noindent Viewing $F (W,\varpi) = F \left( X(W) , \xi(\varpi) , Y(W) , \eta(\varpi) \right)$ as a function of the variables $(X,\xi)$ and $(Y,\eta)$, which in turn depend on $(W,\varpi)$ via the coordinate identifications $X=Y=W$ and $\xi=\eta=\varpi$, the differential operators $W\partial_W$ and $\partial_\varpi$ act on $F$ via the multivariable chain rule, yielding
\begin{equation}
W \partial_W  F(W,\varpi) = \left. \left( X \partial_X + Y \partial_Y \right) F(X,\xi,Y,\eta) \right\vert_{\substack{X=Y=W \\ \xi=\eta=\varpi}},
\end{equation}

\noindent and similarly
\begin{equation}
\partial_\varpi F(W,\varpi) = \left. \left( \partial_\xi + \partial_\eta \right) F(X,\xi,Y,\eta) \right\vert_{\substack{X=Y=W \\ \xi=\eta=\varpi}}.
\end{equation}

\noindent Consequently, we find that
\begin{align}
&\left(W\partial_W \partial_\zeta - Z\partial_Z \partial_\varpi\right) F(W,\varpi) \, h(Z,\zeta) \nonumber \\
&\qquad = \left. \left( (X\partial_X + Y\partial_Y)\partial_\zeta - Z\partial_Z (\partial_\xi + \partial_\eta) \right) F(X,\xi,Y,\eta) \, h(Z,\zeta) \right\vert_{\substack{X=Y=W \\ \xi=\eta=\varpi}} \nonumber \\
&\qquad = \left. \left( D_{XZ} + D_{YZ} \right) F(X,\xi,Y,\eta) \, h(Z,\zeta) \right\vert_{\substack{X=Y=W \\ \xi=\eta=\varpi}}.
\label{DiffOpStep}
\end{align}

\noindent Note that Equation \eqref{DiffOpStep} not only holds for $F=(f\star_\Theta g)$, but for arbitrary functions of the variables $(X,\xi)$ and $(Y,\eta)$. Applying the same identity to $(D_{XZ}+D_{YZ})F\,h$ as well as higher powers of the differential operator $(D_{XZ}+D_{YZ})$ acting on $F\,h$, it thus follows by induction that
\begin{align}
&\left(W\partial_W \partial_\zeta - Z\partial_Z \partial_\varpi\right)^k F(W,\varpi) \, h(Z,\zeta) = \left. \left( D_{XZ} + D_{YZ} \right)^k F(X,\xi,Y,\eta)\, h(Z,\zeta) \right\vert_{\substack{X=Y=W \\ \xi=\eta=\varpi}}
\end{align}

\noindent for $k\in\mathbb{N}_0$, which implies that
\begin{equation}
e^{i\theta\left(W\partial_W \partial_\zeta - Z\partial_Z \partial_\varpi \right)} \big(F(W,\varpi) h(Z,\zeta)\big) = \left. e^{i\theta \left( D_{XZ} + D_{YZ} \right)} F(X,\xi,Y,\eta) h(Z,\zeta) \right\vert_{\substack{X=Y=W \\ \xi=\eta=\varpi}}.
\end{equation}

\noindent Altogether, the expression $\big( (f\star_\Theta g) \star_\Theta h \big)$ can therefore be written as
\begin{equation}
\label{TripleExpansionProduct_fg_h_final}
\big( (f\star_\Theta g) \star_\Theta h \big)(V,\varphi) = \left. e^{i\theta\left(D_{XZ} + D_{YZ}\right)} e^{i\theta D_{XY}} f(X,\xi) g(Y,\eta) h(Z,\zeta) \right\vert_{\substack{X=Y=Z=V \\ \xi=\eta=\zeta=\varphi}}.
\end{equation}

\noindent However, if we alternatively write
\begin{equation}
(g\star_\Theta h)(W,\varpi) = \left. e^{i\theta D_{YZ}} g(Y,\eta) h(Z,\zeta) \right\vert_{\substack{Y=Z=W \\ \eta=\zeta=\varpi}},
\end{equation}

\noindent we find in analogy to the previous computation that
\begin{equation}
\label{TripleExpansionProduct_f_gh_final}
\big( f\star_\Theta (g \star_\Theta h) \big)(V,\varphi) = \left. e^{i\theta\left(D_{XY} + D_{XZ}\right)} e^{i\theta D_{YZ}} f(X,\xi) g(Y,\eta) h(Z,\zeta) \right\vert_{\substack{X=Y=Z=V \\ \xi=\eta=\zeta=\varphi}}.
\end{equation}

\noindent Recalling that the generating vector fields $V\partial_V$ and $\partial_\varphi$ commute (cf. \cite{Carter:1970kc}), it follows that the differential operators $D_{XY}$, $D_{YZ}$, and $D_{XZ}$ also commute, i.e.,
\begin{equation}
\left[D_{XY} , D_{YZ} \right] = \left[D_{XY} , D_{XZ} \right] = \left[D_{YZ} , D_{XZ} \right] = 0.
\end{equation}

\noindent Using the Baker-Campbell-Hausdorff formula, it therefore follows that
\begin{equation}
e^{i\theta\left(D_{XZ} + D_{YZ}\right)} e^{i\theta D_{XY}} = e^{i\theta\left(D_{XY} + D_{XZ} + D_{YZ}\right)} = e^{i\theta\left(D_{XY} + D_{XZ}\right)} e^{i\theta D_{YZ}},
\end{equation}

\noindent which ensures the equality of \eqref{TripleExpansionProduct_fg_h_final} and \eqref{TripleExpansionProduct_f_gh_final} to each finite order in $\theta$. $\qedsymbol$

\subsection{Proof for Lemma \ref{LemmaDeformedSymplecticInvariance}}\label{ProofDeformedSymplecticInvariance}
Recall that the symplectic form $\sigma$ in the undeformed theory is invariant under $\phi_t$ \cite{SV:1996tomitaki,KPV:2021ea}. We can easily prove this by considering the expression
\begin{align}
\label{UndeformedSymplectic2pitFactor}
\sigma(\phi_t^\ast f, \phi_t^\ast g) &= \int_{\mathbb{R}_+\times\mathcal{S}} \left( e^{-2\pi t}(\partial_V f)(e^{-2\pi t}V) g(e^{2\pi t}V) - e^{-2\pi t}(\partial_V g)(e^{-2\pi t}V)f(e^{-2\pi t}V) \right) dV  d\mathrm{vol}_\mathcal{S},
\end{align}

\noindent where we have used the identity $\partial_V(\phi_t^\ast f) = e^{-2\pi t} \phi_t^\ast (\partial_V f)$, and applying a change of variables $V' = e^{-2\pi t} V$, which yields that $dV = e^{2\pi t} dV'$, leading to
\begin{align}
\sigma(\phi_t^\ast f, \phi_t^\ast g) &= \int_{\mathbb{R}_+\times\mathcal{S}} \big((\partial_{V'} f)(V') g(V') - (\partial_{V'} g)(V') f(V') \big) \, dV' d\mathrm{vol}_\mathcal{S} \nonumber \\
&= \sigma(f,g).
\end{align}

\noindent Moreover, we recall that the deformed product $\star_\Theta$ is generated by the actions of dilations and rotations, which commute with the projected Killing flow, and so do their generators, i.e.,
\begin{align}
V \partial_V (\phi_t^\ast f) &= \phi_t^\ast \left( V\partial_V f \right), \\
\partial_\varphi (\phi_t^\ast f) &= \phi_t^\ast (\partial_\varphi f),
\end{align}

\noindent for all $f\in\mathscr{D}_{\mathcal{H}_A}$. Consequently, we employ the truncated formal power series
\begin{equation}
\left. e^{i\theta\left(X\partial_X \partial_\eta - Y\partial_Y \partial_\xi \right)} \right\vert^{(N)} := \sum_{k=0}^N\frac{(i\theta)^k}{k!} \left(X\partial_X \partial_\eta -Y\partial_Y \partial_\xi\right)^k,
\end{equation}

\noindent to compute the dilated deformed symplectic form up to arbitrary finite order $N$ in $\theta$, i.e.,
\begin{align}
\sigma_\Theta^{(N)}(\phi_t^\ast f, \phi_t^\ast g)&=\int_{\mathbb{R}_+\times\mathcal{S}} e^{i\theta\left(X\partial_X \partial_\eta - Y\partial_Y \partial_\xi \right)} \big( \partial_X \phi_t^\ast f(X,\vartheta,\xi) \, \phi_t^\ast g(Y,\vartheta,\eta) \\
&\hspace{4.7cm}- \partial_X \phi_t^\ast g(X,\vartheta,\xi) \, \phi_t^\ast f(Y,\vartheta,\eta)\big) \bigg\vert_{\substack{X=Y=V \\ \xi=\eta=\varphi}}^{(N)} \,dV d\mathrm{vol}_\mathcal{S} \nonumber \\
&=\int_{\mathbb{R}_+\times\mathcal{S}} e^{i\theta\left(X\partial_X \partial_\eta - Y\partial_Y \partial_\xi \right)} e^{-2\pi t} \phi_t^\ast \big( \partial_X f(X,\vartheta,\xi) \, g(Y,\vartheta,\eta) \nonumber \\
&\hspace{5.6cm}-  \partial_X g(X,\vartheta,\xi) \, f(Y,\vartheta,\eta)\big) \bigg\vert_{\substack{X=Y=V \\ \xi=\eta=\varphi}}^{(N)} \,dV d\mathrm{vol}_\mathcal{S} \nonumber
\end{align}

\noindent where the global factor $e^{-2\pi t}$ appears analogously to Equation \eqref{UndeformedSymplectic2pitFactor}. Moreover, using that $\phi_t^\ast$ commutes with all finite powers of $X\partial_X$, $Y\partial_Y$, $\partial_\xi$, and $\partial_\eta$, we obtain that
\begin{align}
\sigma_\Theta^{(N)}(\phi_t^\ast f, \phi_t^\ast g)&=\int_{\mathbb{R}_+\times\mathcal{S}} \phi_t^\ast \, e^{i\theta\left(X\partial_X \partial_\eta - Y\partial_Y \partial_\xi \right)} \big( \partial_X f(X,\vartheta,\xi) \, g(Y,\vartheta,\eta) \\
&\hspace{4.65cm}- \partial_X g(X,\vartheta,\xi) \, f(Y,\vartheta,\eta)\big) \bigg\vert_{\substack{X=Y=V \\ \xi=\eta=\varphi}}^{(N)} \, e^{-2\pi t} dV d\mathrm{vol}_\mathcal{S} \nonumber
\end{align}

\noindent in the sense of formal power series. Acting $\phi_t^\ast$ on each order of the expansion, and, employing the variable changes $X' = e^{-2\pi t} X$ and $Y' = e^{-2\pi t} Y$ (corresponding to the previous relation $V' = e^{-2\pi t} V$), we arrive at
\begin{align}
\sigma_\Theta^{(N)}(\phi_t^\ast f, \phi_t^\ast g)&=\int_{\mathbb{R}_+\times\mathcal{S}} e^{i\theta\left(X'\partial_{X'} \partial_\eta - Y'\partial_{Y'} \partial_\xi \right)} \big( \partial_{X'}f(X',\vartheta,\xi) \, g(Y',\vartheta,\eta) \nonumber \\
&\hspace{4.65cm}- \partial_{X'} g(X',\vartheta,\xi) \, f(Y',\vartheta,\eta)\big) \bigg\vert_{\substack{X'=Y'=V' \\ \xi=\eta=\varphi}}^{(N)} \, dV' d\mathrm{vol}_\mathcal{S} \nonumber \\
&= \sigma_\Theta^{(N)} (f,g).
\end{align}

\noindent $\qedsymbol$

\subsection{Proof for Corollary \ref{CorollarySecondOrderEntropy}}\label{ProofSecondOrderEntropy}
Considering the relative entropy \eqref{DeformedRelativeEntropy} in the deformed theory, we begin with the \emph{first-order} expansion\footnote{For the sake of readability, we will discuss the second-order terms separately afterwards.} according to Formula \eqref{FirstOrderProduct}, for which we obtain
\begin{align}
S^\mathrm{rel} \left( \mathring{\omega}_\Theta \, \Vert \, \omega^f_\Theta \right) &= \pi \int_{\mathbb{R}_+\times\mathcal{S}} \left( (\partial_V f) \star_\Theta (V(\partial_V f)) -  \left((\partial_V f) +(V(\partial^2_V f)) \right) \star_\Theta f \right) \, dV d\mathrm{vol}_\mathcal{S} \nonumber \\
&= \pi \int_{\mathbb{R}_+\times\mathcal{S}} \left( V(\partial_V f)^2 - f (\partial_V f) - V(\partial^2_V f) f + i\theta \left\lbrace \partial_V f , V(\partial_V f) \right\rbrace_\Theta \right. \nonumber \\
&\hspace{3.7cm} \left. - i\theta \left\lbrace (\partial_V f) +(V(\partial^2_V f)) , f \right  \rbrace_\Theta \right) dV\,d\mathrm{vol}_\mathcal{S} + \mathcal{O}(\theta^2).
\end{align}

\noindent In analogy to the computations presented in \cite{KPV:2021ea} (see also \cite{Longo:2019relent,Hollands:2020relent,DM:2025sce}), the zeroth-order terms can be simplified via partial integration, or more specifically by
\begin{align}
\label{FirstOrderEntropyStep1}
\int_{\mathbb{R}_+\times\mathcal{S}} (\partial_V f) f dV d\mathrm{vol}_\mathcal{S} = \underbrace{\left. \int_\mathcal{S} f^2d\mathrm{vol}_\mathcal{S} \right\vert_{V\in\mathbb{R}}}_{\longrightarrow 0} - \int_{\mathbb{R}_+\times\mathcal{S}} f (\partial_V f) dV d\mathrm{vol}_\mathcal{S} = 0
\end{align}

\noindent and
\begin{align}
\int_{\mathbb{R}_+\times\mathcal{S}} V (\partial_V^2 f) f dV d\mathrm{vol}_\mathcal{S} &= \underbrace{\left. \int_\mathcal{S} V f (\partial_V f) d\mathrm{vol}_\mathcal{S} \right\vert_{V\in\mathbb{R}}}_{\longrightarrow 0} - \int_{\mathbb{R}_+\times\mathcal{S}} (\partial_V f)\left( f + V (\partial_V f) \right) \,dV d\mathrm{vol}_\mathcal{S} \nonumber \\
&= - \underbrace{\int_{\mathbb{R}_+\times\mathcal{S}} f (\partial_V f) dV d\mathrm{vol}_\mathcal{S}}_{\longrightarrow0\; \text{by \eqref{FirstOrderEntropyStep1}}} - \int_{\mathbb{R}_+\times\mathcal{S}} V (\partial_V f)^2 dV d\mathrm{vol}_\mathcal{S},
\end{align}

\noindent in order to obtain
\begin{align}
S^\mathrm{rel} \left( \mathring{\omega}_\Theta \, \Vert \, \omega^f_\Theta \right) &= 2\pi \int_{\mathbb{R}_+\times\mathcal{S}} V(\partial_V f)^2 dV\,d\mathrm{vol}_\mathcal{S} \\
&\hspace{0.5cm} + i\pi\theta \int_{\mathbb{R}_+\times\mathcal{S}} \left( \left\lbrace \partial_V f , V(\partial_V f) \right\rbrace_\Theta  - \left\lbrace (\partial_V f) +(V(\partial^2_V f)) , f \right\rbrace_\Theta \right) \,dV d\mathrm{vol}_\mathcal{S} \nonumber \\
&\hspace{0.5cm} + \mathcal{O}(\theta^2).
\end{align}

\noindent Primarily, we notice that the zeroth-order terms reduce to the relative entropy between coherent excitations in the undeformed theory on $\mathcal{H}_A^\mathcal{R} \cong \mathbb{R}_+\times\mathcal{S}$. In order to compute the first-order contributions in $\theta$, we explicitly consider the Poisson brackets
\begin{align}
\left\lbrace \partial_V f , V(\partial_V f) \right\rbrace_\Theta &= V(\partial_V^2 f)(\partial_\varphi V \partial_V f) - V\partial_V (V\partial_Vf)(\partial_\varphi\partial_Vf) \nonumber \\
&= V(\partial_V^2f)(V\partial_\varphi\partial_Vf) - V(\partial_V f + V\partial_V^2f)(\partial_\varphi\partial_Vf) \nonumber\\
&= - V (\partial_V f) (\partial_\varphi \partial_V f)
\end{align}

\noindent and
\begin{align}
\left\lbrace (\partial_V f) +(V(\partial^2_V f)),f \right\rbrace_\Theta &= \left( V\partial_V (\partial_V f) + V\partial_V\left(V(\partial_V^2f)\right) \right) (\partial_\varphi f) \nonumber \\
&\hspace{0,5cm} - V \partial_V f \left(\partial_\varphi\partial_Vf+\partial_\varphi(V\partial_V^2f)\right) \nonumber \\
&=  V(\partial^2_V f)(\partial_\varphi f) + V\left( \partial_V^2 f + V(\partial^3_V f) \right)(\partial_\varphi f) \nonumber \\
&\hspace{0,5cm} - V (\partial_Vf) (\partial_\varphi\partial_V f) - V^2 (\partial_Vf)(\partial_\varphi\partial^2_V f) \nonumber \\
&= 2V(\partial_V^2 f)(\partial_\varphi f) + V^2(\partial_V^3 f)(\partial_\varphi f) \nonumber\\
&\hspace{0,5cm} - V(\partial_V f)(\partial_\varphi\partial_V f) - V^2 (\partial_Vf)(\partial_\varphi\partial^2_V f) .
\end{align}

\noindent Thus, the relative entropy in the deformed theory takes the form
\begin{align}
\label{FirstOrderEntropyStep2}
S^\mathrm{rel} \left( \mathring{\omega}_\Theta \, \Vert \, \omega^f_\Theta \right) &= S^\mathrm{rel}(\mathring{\omega}\,\Vert\,\omega^f) \\
&\hspace{0,5cm} + i\theta\pi \int_{\mathbb{R}_+\times\mathcal{S}} \left( 2V(\partial_V^2f)(\partial_\varphi f) + V^2 (\partial_V^3 f)(\partial_\varphi f) - V^2(\partial_Vf)(\partial_\varphi\partial_V^2f) \right) \,dV d\mathrm{vol}_\mathcal{S}. \nonumber 
\end{align}

\noindent Furthermore, we may integrate by parts once more, so that we first obtain
\begin{align}
\int_{\mathbb{R}_+\times\mathcal{S}} \left( V^2(\partial_Vf)(\partial_\varphi\partial_V^2f) \right) \,dV d\mathrm{vol}_\mathcal{S} &= \underbrace{\left. \int_\mathcal{S} V^2 f (\partial_\varphi\partial_V^2 f) d\mathrm{vol}_\mathcal{S}\right\vert_{V\in\mathbb{R}}}_{\longrightarrow0} \\
&\hspace{0.5cm}- \int_{\mathbb{R}_+\times\mathcal{S}} f \left( 2V (\partial_\varphi\partial_V^2 f) + V^2 (\partial_\varphi \partial_V^3 f) \right) dV d\mathrm{vol}_\mathcal{S}. \nonumber
\end{align}

\noindent by integrating with respect to $V$ over $\mathbb{R}$, and
\begin{align}
\int_{\mathbb{R}_+\times\mathcal{S}} f \,\partial_\varphi\left( 2V (\partial_V^2 f) + V^2 (\partial_V^3 f) \right) dV d\mathrm{vol}_\mathcal{S} &= \underbrace{\left. \int_{\mathbb{R}_+\times[0,\pi]} f \left( 2V (\partial_V^2 f) + V^2 (\partial_V^3 f) \right) dV\rho(\vartheta) d\vartheta \right\vert_{-\pi}^{\pi}}_{\longrightarrow 0 \; \text{since $f\in\mathscr{D}_{\mathcal{H}_A}$}} \nonumber \\
&\hspace{0.5cm} - \int_{\mathbb{R}_+\times\mathcal{S}} (\partial_\varphi f)\left( 2V (\partial_V^2 f) + V^2 (\partial_V^3 f) \right) dV d\mathrm{vol}_\mathcal{S}.
\end{align}

\noindent by integrating with respect to $\varphi$ over the interval $(-\pi,\pi)$, where we have used that the volume element $d\mathrm{vol}_\mathcal{S}$ can be written as $\rho(\vartheta)\, d\vartheta \,d\varphi$ for some suitable function $\rho$ of $\vartheta$, and that $f(V,\vartheta,-\pi) = 0 =f(V,\vartheta,\pi)$ for all $f\in\mathscr{D}_{\mathcal{H}_A}$. Hence, we observe that the first-order terms in $\theta$, as denoted in expression \eqref{FirstOrderEntropyStep2}, cancel each other, so that the first-order contribution vanishes altogether. \\

Next, we consider the second-order contributions of the deformed product individually, using Formula \eqref{SecondOrderProduct}. In a similar fashion to the previous calculations, we explicitly obtain the second-order integrand
\begin{align}
I\left(\theta^2\right) &= -\frac{\pi\theta^2}{2} \big( V^2 \partial_V^2 \left( V\partial_V f \right) (\partial_\varphi^2 \partial_V f) + V\partial_V(V\partial_V f)(\partial_\varphi^2 \partial_V f) + V^2 (\partial_V^3 f) \partial_\varphi^2(V\partial_V f) \nonumber \\
&\hspace{1.67cm} +V\partial_V(\partial_V f) \, \partial_\varphi^2 (V\partial_V f) - 2V^2 \partial_\varphi \partial_V (V\partial_V f) (\partial_\varphi \partial_V^2 f) \big) \nonumber \\
&\hspace{0.5cm} +\frac{\pi\theta^2}{2} \big( V^2 (\partial_V^2 f) \partial_\varphi^2 \left( \partial_Vf + V \partial_V^2 f \right) + (V\partial_V f)\partial_\varphi^2 (\partial_V f + V\partial_V^2 f) \nonumber \\
&\hspace{1.67cm} + V^2 \partial_V^2 \left( \partial_V f + V\partial_V^2 f \right)(\partial_\varphi^2 f) + V\partial_V (\partial_V f + V \partial_V^2 f) (\partial_\varphi^2 f) \nonumber \\
&\hspace{1.67cm} - 2V^2 (\partial_\varphi \partial_V f) \partial_\varphi\partial_V \left( \partial_V f + V\partial_V^2 f \right) \big) \nonumber \\
%&= \frac{\pi\theta^2}{2} \big( V^2 (\partial_V^2 f)(\partial_\varphi^2 \partial_V f) + V^3 (\partial_V^2 f)(\partial_V^2 \partial_\varphi^2f) + 3V^2 (\partial_V^3 f)(\partial_\varphi^2 f) + V^3 (\partial_V^4 f)(\partial_\varphi^2 f) \nonumber \\
%&\hspace{1.2cm} - 4V^2 (\partial_\varphi \partial_V f)(\partial_\varphi \partial_V^2 f) - 2V^3 (\partial_\varphi \partial_V f)(\partial_V^3 \partial_\varphi f) - 2V^2 (\partial_V^2f) (\partial_V \partial_\varphi^2 f) \nonumber \\
%&\hspace{1.2cm} - 2V^3 (\partial_V^3f)(\partial_V\partial_\varphi^2 f) + 2V^2 (\partial_V \partial_\varphi f)(\partial_\varphi\partial_V^2f) + 2V^3 (\partial_\varphi \partial_V^2 f)^2 - (V\partial_V f)(\partial_V\partial_\varphi^2 f) \nonumber \\
%&\hspace{1.2cm} + (V\partial_V f)(\partial_V\partial_\varphi^2 f) - 2V^2 (\partial_V^2 f)(\partial_V \partial_\varphi^2 f) + V^2 (\partial_V f)(\partial_V^2 \partial_\varphi^2 f) + 2V (\partial_V^2 f)(\partial_\varphi^2 f) \nonumber \\
%&\hspace{1.2cm} + V^2 (\partial_V^3 f)(\partial_\varphi^2 f) \big) \nonumber \\
&= \frac{\pi\theta^2}{2} \left( V^3 (\partial_V^2 f)(\partial_V^2 \partial_\varphi^2 f) + 2V^3 (\partial_\varphi \partial_V^2 f)^2 + 3V^2 (\partial_V^3 f)(\partial_\varphi^2 f) + V^3 (\partial_V^4 f) (\partial_\varphi^2 f) \right. \nonumber \\
&\hspace{1.2cm} + V^2 (\partial_V^3 f)(\partial_\varphi^2 f) + 2V (\partial_V^2 f)(\partial_\varphi^2 f) + V^2 (\partial_V f)(\partial_V^2 \partial_\varphi^2 f) - V^2 (\partial_V^2 f)(\partial_\varphi^2 \partial_V f) \nonumber \\
&\hspace{1.2cm} - 2V^2 (\partial_\varphi \partial_V f)(\partial_\varphi\partial_V^2 f)
- 2V^3 (\partial_\varphi \partial_V f) (\partial_V^3 \partial_\varphi f) - 2V^3 (\partial_V^3 f)(\partial_V \partial_\varphi^2 f)  \nonumber\\
&\hspace{1.2cm} - 2V^2 (\partial_V^2 f)(\partial_V \partial_\varphi^2 f) \big).
\label{ComplicatedExpression}
\end{align}

\noindent We observe that in this seemingly complicated expression, most terms are pairwise related via partial integration. In particular, we have

\begin{align}
\int_{\mathbb{R}_+\times\mathcal{S}} 2V^j (\partial_\varphi \partial_V f)(\partial_\varphi \partial_V^j f) dV d\mathrm{vol}_\mathcal{S} &= \underbrace{\left. \int_{\mathbb{R}_+\times[0,\pi]} 2V^j (\partial_\varphi \partial_V f)(\partial_V^j f) \, dV \rho(\vartheta)\, d\vartheta  \right\vert_{\varphi=-\pi}^{\pi}}_{\longrightarrow 0} \nonumber \\
&\hspace{0.5cm} - \int_{\mathbb{R}_+\times\mathcal{S}} 2V^j (\partial_\varphi^2 \partial_V f)(\partial_V^j f)\,  dV d\mathrm{vol}_\mathcal{S}
\end{align}

\noindent for $j=2,3$,
\begin{align}
\int_{\mathbb{R}_+\times\mathcal{S}} V^3 (\partial^2_V f)(\partial_\varphi^2 \partial_V^2 f) dV d\mathrm{vol}_\mathcal{S} &= \underbrace{\left. \int_{\mathbb{R}_+\times[0,\pi]}  V^3 (\partial_V^2 f)(\partial_\varphi\partial_V^2 f) \,dV \rho(\vartheta)\, d\vartheta  \right\vert_{\varphi=-\pi}^{\pi}}_{\longrightarrow 0} \nonumber \\
&\hspace{0.5cm} - \int_{\mathbb{R}_+\times\mathcal{S}} V^3 (\partial_\varphi \partial_V^2 f)^2 dV d\mathrm{vol}_\mathcal{S},
\end{align}

\noindent and
\begin{align}
\int_{\mathbb{R}_+\times\mathcal{S}} (3V^2)(\partial_V^3 f)(\partial_\varphi^2f) \,dV d\mathrm{vol}_\mathcal{S} &= \underbrace{\left. \int_\mathcal{S} V^3 (\partial_V^3f) (\partial_\varphi^2 f) d\mathrm{vol} \right\vert_{V\in\mathbb{R}}}_{\longrightarrow0} \\
&\hspace{0.5cm}- \int_{\mathbb{R}_+\times\mathcal{S}} V^3 \left( (\partial_V^4 f)(\partial_\varphi^2 f) + (\partial_V^3f)(\partial_\varphi^2 \partial_Vf) \right) dV d\mathrm{vol}_\mathcal{S}. \nonumber
\end{align}

\noindent Similarly, we find that
\begin{align}
\int_{\mathbb{R}_+\times\mathcal{S}} V^2 (\partial_V^3 f) (\partial_\varphi^2f) \, dV d\mathrm{vol}_\mathcal{S} = - \int_{\mathbb{R}_+\times\mathcal{S}} \left( V^2 (\partial_V^2 f)(\partial_V \partial_\varphi^2 f) + 2V (\partial_V^2 f) (\partial_\varphi^2 f) \right) dV d\mathrm{vol}_\mathcal{S},
\end{align}

\noindent as well as
\begin{align}
\int_{\mathbb{R}_+\times\mathcal{S}} V^2 (\partial_V f) (\partial_V^2 \partial_\varphi^2 f) \, dV d\mathrm{vol}_\mathcal{S} = - \int_{\mathbb{R}_+\times\mathcal{S}} \left( V^2 (\partial_V^2 f)(\partial_V \partial_\varphi^2 f) + 2V (\partial_V f) (\partial_V \partial_\varphi^2 f)  \right)dV d\mathrm{vol}_\mathcal{S}.
\end{align}

\noindent Integrating Expression \eqref{ComplicatedExpression} over $\mathbb{R}_+\times\mathcal{S}$ and performing a few simplifications, we therefore obtain the second-order expression
\begin{align}
\mathcal{O}(\theta^2) &= \frac{\pi\theta^2}{2} \int_{\mathbb{R}_+\times\mathcal{S}} \big( V^3(\partial_\varphi \partial_V^2 f)^2 - V^3(\partial_V^3 f)(\partial_\varphi^2\partial_V f) - 2V (\partial_V f)(\partial_V \partial_\varphi^2 f) \nonumber \\
&\hspace{2.5cm} -3V^2 (\partial_V^2 f)(\partial_V \partial_\varphi^2f) \big)\, dV d\mathrm{vol}_\mathcal{S} + \mathcal{O}(\theta^3).
\label{SemiFinalExpression}
\end{align}

\noindent The final steps consist of integration by parts once more, i.e., we compute
\begin{align}
\label{PartialIntegrationV2}
\int_{\mathbb{R}_+\times\mathcal{S}} V^2 (\partial_V^2 f) (\partial_\varphi^2 \partial_V f) dV d\mathrm{vol}_\mathcal{S} &= -\int_{\mathbb{R}_+\times\mathcal{S}} V^2 (\partial_\varphi \partial_V^2 f) (\partial_V \partial_\varphi f) dV d\mathrm{vol}_\mathcal{S} \\
&= \int_{\mathbb{R}_+\times\mathcal{S}} \left( 2V (\partial_V \partial_\varphi f)^2 +V^2(\partial_\varphi \partial_V f)(\partial_V^2 \partial_\varphi f) \right) dV d\mathrm{vol}_\mathcal{S}, \nonumber
\end{align}

\noindent which directly implies that
\begin{equation}
\label{PartialIntegrationV3}
-2\int_{\mathbb{R}_+\times\mathcal{S}} V^2 (\partial_\varphi \partial_V^2 f) (\partial_V \partial_\varphi f) dV d\mathrm{vol}_\mathcal{S} = 2\int_{\mathbb{R}_+\times\mathcal{S}} V (\partial_V \partial_\varphi f)^2  dV d\mathrm{vol}_\mathcal{S},
\end{equation}

\noindent and hence by Equation \eqref{PartialIntegrationV2}
\begin{equation}
\int_{\mathbb{R}_+\times\mathcal{S}} V^2 (\partial_V^2 f) (\partial_\varphi^2 \partial_V f) dV d\mathrm{vol}_\mathcal{S} = \int_{\mathbb{R}_+\times\mathcal{S}} V (\partial_V \partial_\varphi f)^2  dV d\mathrm{vol}_\mathcal{S}.
\end{equation}

\noindent Similarly, we compute
\begin{align}
-\int_{\mathbb{R}_+\times\mathcal{S}} V^3 (\partial_V^3f)(\partial_\varphi^2\partial_Vf) dV d\mathrm{vol}_\mathcal{S} &= \int_{\mathbb{R}_+\times\mathcal{S}} V^3 (\partial_\varphi\partial_V^3 f)(\partial_\varphi \partial_V f)dV d\mathrm{vol}_\mathcal{S} \\
&= -\int_{\mathbb{R}_+\times\mathcal{S}} 3V^2 (\partial_\varphi \partial_V^2 f)(\partial_\varphi \partial_V f) dV d\mathrm{vol}_\mathcal{S} \nonumber \\
&\hspace{2cm} -\int_{\mathbb{R}_+\times\mathcal{S}} V^3 (\partial_\varphi \partial_V^2 f)^2 dV d\mathrm{vol}_\mathcal{S},
\end{align}

\noindent which by Equation \eqref{PartialIntegrationV3} leads to
\begin{equation}
-\int_{\mathbb{R}_+\times\mathcal{S}} V^3 (\partial_V^3f)(\partial_\varphi^2\partial_Vf) dV d\mathrm{vol}_\mathcal{S} = \int_{\mathbb{R}_+\times\mathcal{S}} \left( 3V (\partial_V \partial_\varphi f)^2 - V^3(\partial_\varphi \partial_V^2f)^2 \right) dV d\mathrm{vol}_\mathcal{S}.
\end{equation}

\noindent Hence, all terms in the integrand of Expression \eqref{SemiFinalExpression} cancel each other, except for
\begin{align}
\mathcal{O}(\theta^2) = -\frac{\pi\theta^2}{2} \int_{\mathbb{R}_+\times\mathcal{S}} 2V \left( \partial_V f \right)\left( \partial_V \partial_\varphi^2 f \right) \, dV d\mathrm{vol}_\mathcal{S} + \mathcal{O}(\theta^3),
\end{align}

\noindent which reduces to
\begin{align}
\mathcal{O}(\theta^2) = \pi\theta^2 \int_{\mathbb{R}_+\times\mathcal{S}} V \left( \partial_V \partial_\varphi f \right)^2 \, dV d\mathrm{vol}_\mathcal{S} + \mathcal{O}(\theta^3).
\end{align}

\noindent upon integration by parts with respect to $\varphi$. $\qedsymbol$

\section{Transversal Decomposition of Null-Surface Algebras}\label{AppendixMTW}
In this appendix, we briefly discuss a few important aspects from \cite{MTW:2022npa} which support the heuristic argument that, for free fields, the algebra associated with a null hypersurface admits a continuous transverse decomposition into identical chiral CFTs propagating along the respective null geodesic generators, parametrized by points on a transverse cross-section. In \cite{MTW:2022npa}, this property is established rigorously for a free scalar field restricted to a null plane in $D$-dimensional Minkowski spacetime, and we employ it as a reference framework for drawing an analogy with the undeformed theory on a Killing horizon discussed in Section \ref{QFTonKillingHorizons}. \\

Most importantly, the one-particle Hilbert space $\mathscr{H}_1$ of the free scalar field on the null plane admits a transverse direct-integral decomposition (cf. \cite[Chapter IV.8]{Takesaki:1979oa})
\begin{equation}
\mathscr{H}_1 \cong \int_{\mathbb{R}^{D-1}}^{\oplus}\mathscr{H}_{x_\perp} \, d^{D-1}x_\perp ,
\end{equation}

\noindent where each $\mathscr{H}_{x_\perp}$ is, up to unitary equivalence, the one-particle Hilbert space of a chiral CFT on the null ray emanating from the point $x_\perp \in \mathbb{R}^{D-1}$ \cite{MTW:2022npa}. Likewise, operators on $\mathscr{H}_1$ that commute with all multiplication operators in the transverse variable are decomposable with respect to this direct-integral structure, i.e., they admit representations of the form $O \cong \int_{\mathbb{R}^{D-1}}^{\oplus} O_{x_\perp}\, d^{D-1}x_\perp$ with $O_{x_\perp}$ acting on $\mathscr{H}_{x_\perp}$ \cite{MTW:2022npa}. \\

Second quantization then lifts this decomposability property from the one-particle Hilbert space $\mathscr{H}_1$ to the corresponding von Neumann algebras on the bosonic Fock space \cite{MTW:2022npa}, so that, in an appropriate representation, the null-surface algebra is generated by operators that are decomposable with respect to the same transverse direct-integral structure. Accordingly, the direct-integral decomposition of the horizon algebra along the transverse direction employed in Section \ref{QFTonKillingHorizons} needs to be understood in the representation-theoretic sense, namely that, upon a suitable representation, the von Neumann algebra is generated by decomposable operators on a direct-integral Hilbert space. \\

Analogously, given a state $\omega$ on the null-surface von Neumann algebra $\mathscr{N}$, its action on decomposable observables $A \cong \int_{\mathbb{R}^{D-1}}^{\oplus} A_{x_\perp}\, d^{D-1}x_\perp$ can be written in the corresponding direct-integral decomposition
\begin{equation}
\omega(A) \cong \int_{\mathbb{R}^{D-1}}^\oplus \omega_{x_\perp}(A_{x_\perp})\, d^{D-1}x_\perp ,
\end{equation}

\noindent for $A_{x_\perp}\in\mathscr{N}_{x_\perp}$, where $\omega_{x_\perp}$ denotes the corresponding state on the algebra $\mathscr{N}_{x_\perp}$ associated with the transverse parameter $x_\perp\in\mathbb{R}^{D-1}$ \cite{Takesaki:1979oa}. \\

In the setting relevant for this work, namely that of an axisymmetric Killing horizon, the framework summarized above can be applied by identifying the interval $(-\pi,\pi)\subset\mathbb{R}$ as the transverse space, with $x_\perp=\varphi$, at some fixed coordinate $\vartheta$. In this way, each value of $\varphi$ labels a null geodesic in $\mathcal{H}_A^\mathcal{R}$, as illustrated in Figure \ref{FigureNullCone}. In other words, this identification provides the structural basis for interpreting the horizon algebra $\mathscr{N}_\mathcal{R}$ as decomposable along the transversal direction, as described in Section \ref{QFTonKillingHorizons}, reflecting a geometric decomposition of the horizon into a congruence of null geodesics.

{\footnotesize
\bibliographystyle{unsrt}
\bibliography{bibliography}}

\end{document}